\documentclass[10pt,reqno]{amsart}
\usepackage{latexsym,amsmath}
\usepackage{times}
\usepackage{amsfonts}
\usepackage{amssymb}
\usepackage{latexsym}
\usepackage{color}
\usepackage[normalem]{ulem}

\usepackage{bbm}

\newenvironment{claimproof}[1]{\noindent{\bf Proof of Claim #1:\@}}{\hfill $\square$\\}

\newcommand\nix{\,\cdot\,}
\newcommand\normA{\norm{\rho_{\nix\star}}_2}
\newcommand\normB{\norm{\rho_{\star\nix}}_2}
\newcommand\normC{\norm{\rho}_2}

\newcommand\Forb{\cF}

\newcommand\contig{\triangleleft}
\newcommand\prc{\pi^{\mathrm{rc}}_{k,n,m}}
\newcommand\ppl{\pi^{\mathrm{pl}}_{k,n,m}}

\newcommand\Bal{\cB_{n,k}(\omega)}

\newcommand{\beq}{\begin{equation}} \newcommand{\eeq}{\end{equation}}

\newcommand\bal{{\mathrm{bal}}}

\newcommand\G{\vec G}

\newcommand\gnp{G(n,p)}

\newcommand\gnm{\cG(n,m)}

\numberwithin{equation}{section}

\def\vec#1{\mbox{\boldmath$#1$}}

\newcommand{\Zkc}{Z_{k}}

\newcommand{\Ztame}{\widetilde Z_{k,\omega}}
\newcommand{\Zbal}{Z_{k,\omega}}

\newcommand{\dc}{d_{k,\mathrm{cond}}}

\newcommand{\dk}{d_{k-\mathrm{col}}}

\newcommand\MPCPS{Mathematical Proceedings of the Cambridge Philosophical Society}
\newcommand\IPL{Information Processing Letters}
\newcommand\COMB{Combinatorica}

\DeclareMathOperator{\pr}{\mathbb P}

\newcommand\SIGMA{\vec\sigma}

\newcommand\TAU{\vec\tau}

\newtheorem{definition}{Definition}[section]
\newtheorem{claim}[definition]{Claim}

\newtheorem{theorem}[definition]{Theorem}
\newtheorem{lemma}[definition]{Lemma}
\newtheorem{proposition}[definition]{Proposition}
\newtheorem{corollary}[definition]{Corollary}

\newtheorem{fact}[definition]{Fact}

\usepackage{fullpage}

\newcommand\cA{\mathcal{A}}
\newcommand\cB{\mathcal{B}}
\newcommand\cC{\mathcal{C}}

\newcommand\cF{\mathcal{F}}
\newcommand\cG{\mathcal{G}}
\newcommand\cE{\mathcal{E}}

\newcommand\cH{\mathcal{H}}
\newcommand\cS{\mathcal{S}}

\newcommand\cV{\mathcal{V}}

\newcommand\cZ{\mathcal{Z}}

\def\cR{{\mathcal R}}
\def\cC{{\mathcal C}}
\def\cE{{\mathcal E}}

\newcommand\eps{\varepsilon}

\newcommand\Erw{\mathbb{E}}
\newcommand\w{{\omega}}

\newcommand{\vecone}{\vec{1}}

\newcommand{\Po}{{\rm Po}}

\newcommand{\bink}[2] {{{#1}\choose {#2}}}

\newcommand\ra{\rightarrow}

\newcommand\bc[1]{\left({#1}\right)}
\newcommand\cbc[1]{\left\{{#1}\right\}}
\newcommand\bcfr[2]{\bc{\frac{#1}{#2}}}
\newcommand{\bck}[1]{\left\langle{#1}\right\rangle}
\newcommand\brk[1]{\left\lbrack{#1}\right\rbrack}
\newcommand\scal[2]{\bck{{#1},{#2}}}
\newcommand\norm[1]{\left\|{#1}\right\|}

\newcommand\RR{\mathbb{R}}

\newcommand{\whp}{w.h.p.}

\newcommand{\Erdos}{Erd\H{o}s}
\newcommand{\Renyi}{R\'enyi}

\newcommand\Lem{Lemma}
\newcommand\Prop{Proposition}
\newcommand\Thm{Theorem}
\newcommand\Cor{Corollary}
\newcommand\Sec{Section}

\begin{document}

\title{Planting colourings silently}

\author[Victor Bapst, Amin Coja-Oghlan, Charilaos Efthymiou]{Victor Bapst$^*$, Amin Coja-Oghlan$^*$, Charilaos Efthymiou}
\thanks{$^*$The research leading to these results has received funding from the European Research Council under the European Union's Seventh Framework
			Programme (FP/2007-2013) / ERC Grant Agreement n.\ 278857--PTCC}
\date{\today}

\address{Victor Bapst, {\tt bapst@math.uni-frankfurt.de}, Goethe University, Mathematics Institute, 10 Robert Mayer St, Frankfurt 60325, Germany.}

\address{Amin Coja-Oghlan, {\tt acoghlan@math.uni-frankfurt.de}, Goethe University, Mathematics Institute, 10 Robert Mayer St, Frankfurt 60325, Germany.}

\address{Charilaos Efthymiou, {\tt efthymiou@math.uni-frankfurt.de}, Goethe University, Mathematics Institute, 10 Robert Mayer St, Frankfurt 60325, Germany.}

\begin{abstract}
\noindent
Let $k\geq3$ be a fixed integer and let $Z_k(G)$ be the number of $k$-colourings of the graph $G$.
For certain values of the average degree, the random variable $Z_k(G(n,m))$ is known to be concentrated in the sense that
$\frac1n(\ln Z_k(G(n,m))-\ln\Erw[Z_k(G(n,m))])$ converges to $0$ in probability [Achlioptas and Coja-Oghlan: Proc.\ FOCS~2008].
In the present paper we prove a significantly stronger concentration result.
Namely, we show that for a wide range of average degrees,
	$\frac1\omega(\ln Z_k(G(n,m))-\ln\Erw[Z_k(G(n,m))])$ converges to $0$ in probability for {\em any} diverging function $\omega=\omega(n)\ra\infty$.
For $k$ exceeding a certain constant $k_0$ this result covers all average degrees up to the so-called {\em condensation phase transition} $\dc$,
and this is best possible.
As an application, we show that the experiment of choosing a $k$-colouring of the random graph $G(n,m)$
uniformly at random is contiguous with respect to the so-called ``planted model''.

\noindent
\end{abstract}

\maketitle
\section{Introduction}

\noindent

\subsection{Background and motivation}
Let $G(n,m)$ denote the random graph on the vertex set $\brk n=\cbc{1,\ldots,n}$ with precisely $m$ edges.
The study of the graph colouring problem on $G(n,m)$ goes back to the seminal paper of \Erdos\ and \Renyi~\cite{ER}.
A wealth of research has since been devoted to either estimating the typical value of the chromatic number of $G(n,m)$~\cite{AchNaor,BBColor,LuczakColor,Matula},
	its concentration~\cite{AlonKriv,Luczak,ShamirSpencer}, or the problem of colouring random graphs by means of efficient algorithms~\cite{AchMolloy,GMcD,KSud};
	for a more complete survey see~\cite{BB,JLR}.
Some of the methods developed in this line of work have had a wide impact on combinatorics  
	(e.g., the use of martingale tail bounds).

Since the 1990s substantial progress has been made in the case of {\em sparse} random graphs,
	where $m=O(n)$ as $n\ra\infty$.
For instance, Achlioptas and Friedgut~\cite{AchFried} proved that
	for any $k\geq3$ there exists a {\em sharp threshold sequence} $\dk(n)$ such that for any fixed $\eps>0$ the random graph $G(n,m)$ is $k$-colourable \whp\
	if $2m/n<\dk(n)-\eps$, whereas $G(n,m)$ fails to be $k$-colourable \whp\ if $2m/n>\dk(n)+\eps$.
The best current bounds~\cite{Covers,Danny} on $\dk(n)$ show that there is a sequence $(\gamma_k)_{k\geq3}$, $\lim_{k\ra\infty}\gamma_k=0$,
such that
	\begin{equation}\label{eqkcol}
	(2k-1)\ln k-2\ln2-\gamma_k\leq\liminf_{n\ra\infty}\dk(n)\leq\limsup_{n\ra\infty}\dk(n)\leq(2k-1)\ln k-1+\gamma_k.
	\end{equation}

In recent work, to a large extent inspired by predictions from statistical physics~\cite{MM},
 it has emerged that properties of {\em typical} $k$-colourings have a very significant impact both on combinatorial and algorithmic aspects
of the random graph colouring problem. 
To be precise, by a typical $k$-colouring we mean a $k$-colouring of the random graph $G(n,m)$ chosen uniformly at random from the set of all its $k$-colourings
	(provided that this set is non-empty).
Properties of such randomly chosen colourings have been harnessed to study the ``geometry'' of the set of $k$-colourings of a
random graph~\cite{Barriers,Molloy} as well as the nature of correlations
between the colours that different vertices take~\cite{Reconstr}.
In particular, the proofs of the bounds~(\ref{eqkcol}) on $\dk(n)$ exploit structural properties such as the ``clustering''
of the set of $k$-colourings and the emergence of ``frozen variables''.

\subsection{Quiet planting}
The notion of choosing a random colouring of a  random graph $G(n,m)$ can be formalised as follows.
Let $\Lambda_{k,n,m}$ be the set of all pairs $(G,\sigma)$ such that $G$ is a graph on $[n]$
with precisely $m$ edges, and $\sigma$ is a $k$-colouring of $G$.
Further, for a graph $G$ let $\Zkc(G)$ signify the number of $k$-colourings of $G$.
Now, define a probability distribution $\prc(G,\sigma)$, called the {\em random colouring model}, on $\Lambda_{k,n,m}$ by letting
	$$\prc(G,\sigma)=\brk{\Zkc(G)\bink{\bink{n}2}m\pr\brk{G(n,m)\mbox{ is $k$-colourable}}}^{-1}.$$
Perhaps more intuitively, this is the distribution produced by the following experiment.
\begin{description}
\item[RC1] Generate a random graph $\G=G(n,m)$ subject to the condition that $\Zkc(\G)>0$.
\item[RC2] Choose a $k$-colouring $\TAU$ of $\G$ uniformly at random. 
	The result of the experiment is $(\G,\TAU)$.
\end{description}

Since we are going to be interested in values of $m/n$ where $G(n,m)$ is $k$-colourable \whp,
the conditioning in step {\bf RC1} is harmless.
But what turns the direct study of the distribution $\prc$ into a challenge is step {\bf RC2}.
This is illustrated by the fact that the best current algorithms for sampling a $k$-colouring of $G(n,m)$ are known to be efficient only
for average degrees $d<k$~\cite{Efthymiou}, a far cry from $\dk(n)$, cf.~(\ref{eqkcol}).

Achlioptas and Coja-Oghlan~\cite{Barriers} suggested to circumvent this problem by means of an alternative probability distribution
on $\Lambda_{k,n,m}$ called the {\em planted model}.
This distribution is induced by the following experiment;
for $\sigma:[n]\ra[k]$ let
	$$\Forb(\sigma)=\sum_{i=1}^k\bink{|\sigma^{-1}(i)|}2$$
denote the number of edges of the complete graph that are monochromatic under $\sigma$.
\begin{description}
\item[PL1] Choose a map $\SIGMA:\brk n\ra\brk k$ uniformly at random, subject to the condition that $\Forb(\SIGMA)\leq\bink n2- m$.
\item[PL2] Generate a graph $\G$ on $\brk n$ consisting of $m$ edges that are bichromatic under $\SIGMA$ uniformly at random.
	The result of the experiment is $(\G,\SIGMA)$.
\end{description}
Thus, the probability that the planted model assigns to a pair $(G,\sigma)$ is
	$$\ppl(G,\sigma)\sim\brk{\bink{\bink{n}2}m k^n\pr\brk{\mbox{$\sigma$ is a $k$-colouring of $G(n,m)$}}}^{-1}.$$
In contrast to the ``difficult'' experiment {\bf RC1--RC2}, {\bf PL1--PL2} is quite convenient to work with. 

Of course, the two probability distributions $\prc$ and $\ppl$ differ.
For instance, under $\ppl$ a graph $G$ comes up with a probability that is proportional to its number of $k$-colourings,
which is not the case under $\prc$.
However, the two models are related if $m=m(n)$ is such that
	\beq\label{eqFriedgutConcentration}
	\ln\Zkc(G(n,m))=\ln\Erw[\Zkc(G(n,m))]+o(n)\qquad\mbox\whp
	\eeq
Indeed, if~(\ref{eqFriedgutConcentration}) is satisfied, then the following is true~\cite{Barriers}.
	\begin{equation}\label{eqplantingTrick} 
	\parbox{15cm}{If $(\cE_n)$ is a sequence of events $\cE_n\subset\Lambda_{k,n,m}$
		such that $\ppl[\cE_n]\leq\exp(-\Omega(n))$, then $\prc[\cE_n]=o(1)$.}
	\end{equation}
The statement~(\ref{eqplantingTrick}), baptised ``quiet planting'' by Krzalaka and Zdeborov\'a~\cite{QuietPlanting}, has provided the foundation for the study of the geometry of the set
of colourings, freezing etc.~\cite{Barriers,Cond,Molloy,Reconstr}.
Moreover, similar statements have proved useful in the study of other
random constraint satisfaction problems~\cite{Angelica,MolloyRestrepo,Reconstr}.
Yet a significant complication in the use of~(\ref{eqplantingTrick}) is that 
 $\cE_n$ is required to be {\em exponentially} unlikely in the planted model.
This  has caused substantial difficulties in several applications (e.g., \cite{Cond,Molloy}).

\subsection{Results}
The contribution of the present paper is to show that the statement~(\ref{eqplantingTrick}) can be sharpened in the strongest possible sense.
Roughly speaking, we are going to show that if~(\ref{eqFriedgutConcentration}) holds, then the random colouring model is contiguous with respect to the planted model,
	i.e., in~(\ref{eqplantingTrick}) it suffices that $\ppl[\cE_n]=o(1)$ (see \Thm~\ref{Thm_cont} below for a precise statement).
We obtain this result by establishing that under certain conditions the number $\Zkc(G(n,m))$ of $k$-colourings of the random graph
is concentrated remarkably tightly.

To state the result, we need a bit of notation.
From here on out we always assume that $m=\lceil dn/2\rceil$ for a number $d>0$ that remains fixed as $n\ra\infty$.
Furthermore, for $k\geq3$ we define
	\begin{equation}\label{eqdc}
	\dc=\sup\cbc{d>0:\lim_{n\ra\infty}\Erw[\Zkc(G(n,m))^{1/n}]=k(1-1/k)^{d/2}}.
	\end{equation}
This definition is motivated by the well-known fact that
	\beq\label{eqNaiveFirstMoment}
	\Erw[\Zkc(G(n,m))]=\Theta(k^n(1-1/k)^{m}),
	\eeq
Thus, Jensen's inequality shows that $\limsup_{n\ra\infty}\Erw[\Zkc(G(n,m))^{1/n}]\leq k(1-1/k)^{d/2}$ for all $d$,
	and $\dc$ marks the greatest average degree up to which this upper bound is tight.
Under the assumption that $k\geq k_0$ for a certain constant $k_0$
it is possible to calculate the number $\dc$ precisely~\cite{Cond}, and an
asymptotic expansion in $k$ yields
	\begin{align*}
	\dc&=(2k-1)\ln k-2\ln2+\gamma_k,\qquad\mbox{ where }\lim_{k\ra\infty}\gamma_k=0.
	\end{align*}

\begin{theorem}\label{Thm_conc}
There is a constant $k_0>3$ such that the following is true.
Assume either  that $k\geq3$ and $d\leq 2(k-1)\ln(k-1)$ or that $k\geq k_0$ and $d<\dc$.
Then
	\beq\label{eqThm_conc666}
	\lim_{\omega\ra\infty}\lim_{n\ra\infty}\pr\brk{|\ln\Zkc(G(n,m))-\ln\Erw[\Zkc(G(n,m))]|\leq\omega}=1.
	\eeq
On the other hand, for any fixed number $\omega>0$, any $k\geq3$ and any $d>0$ we have
	$$\lim_{n\ra\infty}\pr\brk{|\ln\Zkc(G(n,m))-\ln\Erw[\Zkc(G(n,m))]|\leq\omega}<1.$$
\end{theorem}

For $d,k$ covered by the first part of \Thm~\ref{Thm_conc} we have $\ln\Zkc(G(n,m))=\Theta(n)$ \whp\
Whilst one might expect {\em a priori} that $\ln\Zkc(G(n,m))$ has fluctuations of order, say, $\sqrt n$,
the first part of \Thm~\ref{Thm_conc} shows that actually $\ln\Zkc(G(n,m))$ fluctuates by no more than $\omega(n)$ for {\em any} $\omega(n)\ra\infty$ \whp\
Moreover, the second part shows that this is best possible.
In addition, for $k\geq k_0$ \Thm~\ref{Thm_conc} is best possible with respect to the range of $d$.
In fact, it has been shown in~\cite{Cond} that 
$\ln\Zkc(G(n,m))<\ln\Erw[\Zkc(G(n,m))]-\Omega(n)$ \whp\ for $d>\dc$.

\Thm~\ref{Thm_conc} enables us to establish a very strong connection between the random colouring model and the planted model.
To state this, we recall the following definition.
Suppose that $\vec\mu=(\mu_n)_{n\geq1},\vec\nu=(\nu_n)_{n\geq1}$ are two sequences of probability measures such that
$\mu_n,\nu_n$ are defined on the same probability space $\Omega_n$ for every $n$.
Then $(\mu_n)_{n\geq1}$ is {\em contiguous} with respect to $(\nu_n)_{n\geq1}$, in symbols $\vec\mu\contig\vec\nu$, if
for any sequence $(\cE_n)_{n\geq1}$ of events such that $\lim_{n\ra\infty}\nu_n(\cE_n)=0$ we have $\lim_{n\ra\infty}\mu_n(\cE_n)=0$.

\begin{theorem}\label{Thm_cont}
There is a constant $k_0>3$ such that the following is true.
Assume either  that $k\geq3$ and $d\leq 2(k-1)\ln(k-1)$ or that $k\geq k_0$ and $d<\dc$.
Then 
	$(\prc)_{n\geq1}\contig(\ppl)_{n\geq1}.$
\end{theorem}

\noindent
Inspired by the term ``quiet planting'' that has been used to describe~(\ref{eqplantingTrick}), we are inclined to refer to the contiguity statement 
of \Thm~\ref{Thm_cont} as ``silent planting''.

\subsection{Discussion and further related work.}\label{Sec_discussion}
The proof of \Thm~\ref{Thm_conc} combines the second moment arguments from Achlioptas and Naor~\cite{AchNaor}
and its enhancements from~\cite{Cond,Danny} with the ``small subgraph conditioning'' method~\cite{Janson,RobinsonWormald}.
More precisely, the key observation on which the proof of \Thm~\ref{Thm_conc} is based is that the fluctuations of $\ln\Zkc(G(n,m))$
can be attributed to the variations of the number of bounded length cycles in the random graph.

This was known to be the case in random regular graphs.
In fact, Kemkes, Perez-Gimenez and Wormald~\cite{WormaldColoring} combined the small subgraph conditioning argument 
with the second moment argument from~\cite{AchNaor} to upper-bound the chromatic number of the random $d$-regular graph.
While it had been pointed out by Achlioptas and Moore~\cite{AMoColor} that the second moment argument from~\cite{AchNaor}
can be used rather directly to conclude that the same upper bound 
holds with a probability that remains bounded away from $0$ as $n\ra\infty$, small subgraph conditioning was used in~\cite{WormaldColoring} to boost
this probability to $1-o(1)$. 
Improved bounds on the chromatic number of random regular graphs,
also based on the second moment method and small subgraph conditioning, were recently obtained in~\cite{RegCol}.
In the case of the $G(n,m)$ model, small subgraph conditioning is not necessary to upper-bound
the chromatic number, because
the sharp threshold result~\cite{AchFried} can be used instead. 
	\footnote{While the combination of the second moment method and the sharp threshold result
		can be used to show that (\ref{eqFriedgutConcentration}) implies~(\ref{eqplantingTrick}), this approach
			 does {\em not} 
			 yield \Thm~\ref{Thm_conc}.
		For instance, even the sharp threshold analysis from~\cite{Barriers} allows for the possibility that $\Zkc(G(n,m))=(3-o(1))\Erw[\Zkc(G(n,m))]$ with probability $1/3$,
		while $\Zkc(G(n,m))\leq \exp(-n^{0.99})\Erw[\Zkc(G(n,m))]$ with probability $2/3$.}

A priori it might seem reasonable to expect that the random variable $\ln\Zkc$ is more tightly concentrated in
random regular graphs that in the $G(n,m)$ model, and that therefore small subgraph conditioning cannot be applied in the case of $G(n,m)$.
In fact, in the random regular graph for any fixed number $\omega$ the depth-$\omega$ neighbourhood of all but a bounded number of vertices
is just a $d$-regular tree.
Thus, there are only extremely limited fluctuations in the local structure of the random regular graph.
By contrast, in the $G(n,m)$-model the depth-$\omega$ neighbourhoods can be of varying shapes and sizes
	(although all but a bounded number will be acyclic), and also the number of vertices/edges in the largest connected component and
	the $k$-core fluctuate.
Nonetheless, perhaps somewhat surprisingly, we are going to show that even in the case of the $G(n,m)$ model, the fluctuations of $\ln\Zkc$ are
merely due to the appearance of short cycles.
Finally, \Thm~\ref{Thm_cont} will follow from \Thm~\ref{Thm_conc} by means of a similar argument
as used in~\cite{Barriers}.

We expect that the present approach of combining the second moment method with small subgraph conditioning
can be applied successfully to a variety of other random constraint problems.
Immediate examples that spring to mind include random $k$-NAESAT or random $k$-XORSAT,
random hypergraph $k$-colourability or, more generally, the family of problems studied in~\cite{Reconstr}.
(On the other hand, we expect that in problems such as random $k$-SAT the logarithm of the number of satisfying assignments
exhibits stronger fluctuations, due to a lack of symmetry.)

\subsection{Preliminaries and notation}
We always assume that $n\geq n_0$ is large enough for our various estimates to hold.
Moreover, if $p=(p_1,\ldots,p_l)$ is a vector with entries $p_i\geq0$, then we let
	$$H(p)=-\sum_{i=1}^lp_i\ln p_i.$$
Here and throughout, we use the convention that $0\ln0=0$.
Hence, if $\sum_{i=1}^lp_i=1$, then $H(p)$ is the entropy of the probability distribution $p$.
Further, for a number $x$ and an integer $h>0$ we let $(x)_h=x(x-1)\cdots(x-h+1)$ denote the $h$th falling factorial of $x$.

We use the following instalment of the small subgraph technique.

\begin{theorem}[\cite{Janson}]\label{Thm_Janson}
Suppose that $(\delta_l)_{l \geq2}$, $(\lambda_l)_{l\geq2}$ are sequences of real numbers such that $\delta_l\geq-1$ and $\lambda_l>0$ for all $l$.
Moreover, assume that $(C_{l,n})_{l \geq 2,n\geq1}$ and $(Z_n)_{n\geq1}$ are random variables
such that each $C_{l,n}$ takes values in the non-negative integers. 
Additionally, suppose that  for each $n$ the random variables $C_{2,n},\ldots,C_{n,n}$ and $Z_n$ are defined on the same probability space.
Moreover, 
let $(X_l)_{l\geq2}$ be a sequence of independent random variables such that $X_l$ has distribution $\Po(\lambda_l)$ and
assume that the following four conditions hold.
\begin{description}
\item[SSC1] for any integer $L\geq2$ and any integers $x_2,\ldots,x_L\geq0$ we have
	$$\lim_{n\ra\infty}\pr\brk{\forall 2\leq l\leq L:C_{l,n}=x_l}=\prod_{l=2}^L\pr\brk{X_l=x_l}.$$
\item[SSC2] for any integer $L\geq2$ and any integers $x_2,\ldots,x_L\geq0$ we have
	$$\lim_{n\ra\infty}\frac{\Erw[Z_n|\forall 2\leq l\leq L:C_{l,n}=x_l]}{\Erw[Z_n]}=\prod_{l=2}^L(1+\delta_l)\exp(-\lambda_l\delta_l).$$
\item[SSC3] we have $\sum_{l=2}^\infty\lambda_l\delta_l^2<\infty$.
\item[SSC4] we have
	$\lim_{n\ra\infty}\Erw[Z_n^2]/\Erw[Z_n]^2\leq\exp\brk{\sum_{l=2}^\infty\lambda_l\delta_l^2}. $
\end{description}
Then the sequence $(Z_n/\Erw[Z_n])_{n\geq1}$ converges in distribution to
	$\prod_{l=2}^\infty(1+\delta_l)^{X_l}\exp(-\lambda_l\delta_l).$
\end{theorem}

\section{Outline of the proof}\label{Sec_outline}

\noindent
It turns out to be convenient to prove \Thm s~\ref{Thm_conc} and~\ref{Thm_cont} by way of another random graph model $\gnm$.
This is a random (multi-)graph on the vertex set $[n]$ obtained by choosing $m$ edges $\vec e_1,\ldots,\vec e_m$ of the complete graph on $n$ vertices
uniformly and independently at random (i.e., with replacement).

To bound $\Zkc(\gnm)$ from below,  we will confine ourselves to $k$-colourings in which all the colour classes have very nearly the same size.
More precisely, for a map $\sigma:\brk n\ra\brk k$ we define
	$$\rho(\sigma)=(\rho_1(\sigma),\ldots,\rho_k(\sigma)),\quad\mbox{where }\rho_{i}(\sigma) = | \sigma^{-1}(i) |/n \qquad (i = 1 \dots k).$$
Thus, $\rho(\sigma)$ is a probability distribution on $\brk k$, to which we refer as the {\em colour density} of $\sigma$.
Let $\cC_{k}(n)$ signify the set of all possible colour densities $\rho(\sigma)$, $\sigma:\brk n\ra\brk k$.
Further, let $\overline\cC_k$ be the set of all probability distributions $\rho=(\rho_1,\ldots,\rho_k)$ on $\brk k$,
and let $\rho^\star=(1/k,\ldots,1/k)$ signify the barycentre of $\overline\cC_k$.
We say that $\rho=(\rho_1,\ldots,\rho_k)\in\overline\cC_k$ is {\em $(\omega,n)$-balanced}
if $$|\rho_i-k^{-1}|\leq \omega^{-1}n^{-\frac12}\quad\mbox{ for all $i\in\brk k$.}$$
Let $\cB_{n,k}(\omega)$ denote the set of all $(\omega,n)$-balanced $\rho\in\cC_{k}(n)$.
Now, for a graph $G$ on $[n]$ let
$\Zbal(G)$ signify the number of {\em $(\omega,n)$-balanced $k$-colourings}, i.e., $k$-colourings $\sigma$ such that $\rho(\sigma) \in \cB_{n,k}(\w)$.
In \Sec~\ref{sec_first_moment} we will calculate the first moment of 
$\Zbal$ to obtain the following.

\begin{proposition} \label{prop_first_moment_bal_vanilla}
Fix an integer $k\geq 3$ and a number $d \in (0, \infty)$ and assume that $\omega=\omega(n)$ is a sequence such that $\lim_{n\ra\infty}\omega(n)=\infty$.
Then
	$$
	\Erw \left[\Zkc \right(\gnm)]=\Theta(k^n(1-1/k)^m)\quad\mbox{and}\quad
	\frac{ \Erw \left[ \Zbal \right(\gnm)]}{\Erw \left[\Zkc \right(\gnm)]} \sim \frac{| \Bal |k^{k/2} }{(2 \pi n)^\frac{k-1}{2}}  \left(1+ \frac{d}{k-1} \right)^{\frac{k-1}{2}}. $$
In particular, $ \ln \Erw \left[ \Zbal(\gnm) \right] = \ln \Erw \left[\Zkc(\gnm) \right] + O \left( \ln \omega(n) \right)  $.
\end{proposition}

As outlined in \Sec~\ref{Sec_discussion}, our basic strategy is to show that the fluctuations
of $\Zbal(\gnm)$ can be attributed to fluctuations in the number of cycles of a bounded length.
Hence,  for an integer $l\geq 2$ we let $C_{l,n}$ denote the number of cycles of length (exactly) $l$ in $\gnm$.
Let
	\beq\label{eqlambdadelta}
	\lambda_l=\frac{d^l}{2l}\quad\mbox{ and }\quad\delta_l=\frac{(-1)^l}{(k-1)^{l-1}}.
	\eeq
It is well-known that $C_{2,n},\ldots$ are asymptotically independent Poisson variables (e.g., \cite[\Thm~5.16]{BB}).
More precisely, we have the following.

\begin{fact}\label{Fact_cycles}
If $x_2,\ldots,x_L$ are non-negative integers, then
	$$\lim_{n\ra\infty}\pr\brk{\forall 2\leq l\leq L:C_{l,n}=x_l}=\prod_{l=2}^L\pr\brk{\Po(\lambda_l)=x_l}.$$
\end{fact}

In order to apply \Thm~\ref{Thm_Janson} to the random variables $C_{l,n}$ and $\Zbal(\gnm)$, we
need to investigate the impact of the cycle counts $C_{l,n}$ on the first moment of $\Zbal(\gnm)$.
This is the task that we tackle in \Sec~\ref{sec:prop:CondRatio1stMoment}, where we prove the following.

\begin{proposition}\label{prop:CondRatio1stMoment}
Assume that $k\geq3$ and that $d\in(0,\infty)$. 
Then
	\beq\label{eqlambdadeltaConv}
	\sum_{l=2}^\infty\lambda_l\delta_l^2<\infty.
	\eeq
Moreover, 
let $\omega=\omega(n)>0$ be any sequence such that $\lim_{n\ra\infty}\omega(n)=\infty$.
If $x_2,\ldots, x_L$ are non-negative integers, 
then
\begin{eqnarray}\label{eqprop:CondRatio1stMoment1}
\frac{\mathbb{E}[\Zbal(\gnm)|\forall 2\leq l\leq L:C_{l,n}=x_l]}{\mathbb{E}[\Zbal(\gnm)]}\sim \prod_{l=2}^L\left[1+\delta_l\right]^{x_l}\exp\bc{-\delta_l\lambda_l}.
\end{eqnarray}
\end{proposition}

Additionally, to invoke \Thm~\ref{Thm_Janson} 
we need to know the second moment of $\Zbal(\gnm)$ very precisely.
To obtain the required estimate, we consider two regimes of $d,k$ separately.
In the simpler case, based on the second moment argument from~\cite{AchNaor}, we obtain the following result.

\begin{proposition}\label{prop_second_moment_bal_vanilla}
Assume that $k \geq 3$ and $d < 2(k-1) \ln (k-1)$.
Then
	$$\frac{ \Erw \left[\Zbal(\gnm)^2 \right] }{\Erw \left[ \Zbal(\gnm) \right]^2} \sim \exp \left( \sum_{l \geq 2} \lambda_l \delta_l^2 \right) .  $$
\end{proposition}

\noindent
The second regime of $d,k$ is that $k\geq k_0$ for a certain constant $k_0\geq3$ and $d<\dc$ (with $\dc$ the number defined in~(\ref{eqdc})).
In this case, it is necessary to replace $\Zbal$ by the slightly tweaked random variable $\Ztame$ used
 in the second moment arguments from~\cite{Cond,Danny}.

\begin{proposition}
 \label{prop_first_moment_tame_bal}
There is a constant $k_0\geq3$ such that the following is true.
Assume that $k\geq k_0$ and $2(k-1)\ln(k-1)\leq d<\dc$.
There exists an integer-valued random variable $0\leq\Ztame\leq\Zbal$ such that
	\begin{eqnarray}\label{eqprop_first_moment_tame_bal1}
	\Erw \left[ \Ztame(\gnm) \right] &\sim& \Erw \left[ \Zbal(\gnm) \right]
		\qquad\mbox{and}\\
	\frac{ \Erw \left[ \Ztame(\gnm)^2 \right] }{\Erw \left[ \Ztame(\gnm) \right]^2} &\leq& (1+o(1))\exp \left( \sum_{l \geq 2} \lambda_l \delta_l^2 \right).
		\nonumber
	\end{eqnarray}
\end{proposition}

\noindent
The proofs of \Prop s~\ref{prop_second_moment_bal_vanilla} and~\ref{prop_first_moment_tame_bal} appear at the end of \Sec~\ref{sec_second_moment}.

Of course, to apply \Thm~\ref{Thm_Janson} to the random variable $\Ztame$ we need
to investigate the impact of the cycle counts $C_{l,n}$ on the first moment of $\Ztame$ as well.
That is, we need a similar result as \Prop~\ref{prop:CondRatio1stMoment} for $\Ztame$.
Fortunately, this does not require reiterating the proof of \Prop~\ref{prop:CondRatio1stMoment}.
Instead, what we need follows readily from \Prop~\ref{prop:CondRatio1stMoment} and~(\ref{eqprop_first_moment_tame_bal1}).
More precisely, we have

\begin{corollary}\label{cor_first_moment_tame_bal}
Let $x_2,\ldots,x_L$ be non-negative integers.
With the assumptions and notation of \Prop~\ref{prop_first_moment_tame_bal},
	\begin{equation}\label{eqCor_count3}
	\frac{\mathbb{E}[\Ztame(\gnm)|\forall 2\leq l\leq L:C_{l,n}=x_l]}{\mathbb{E}[\Ztame(\gnm)]}\sim \prod_{l=2}^L\left[1+\delta_l\right]^{x_l}\exp\bc{-\delta_l\lambda_l}.
	\end{equation}
\end{corollary}
\begin{proof}
Let $S$ denote the event $\cbc{\forall l\leq L:C_{l,n}=x_l}$ and let $\cZ_n=\Ztame(\gnm)$ for the sake of brevity.
Since $\cZ_n\leq\Zbal$, (\ref{eqprop_first_moment_tame_bal1}) implies the upper bound
	\begin{eqnarray}\label{eqcor_first_moment_tame_bal12a}
	\frac{\mathbb{E}[\cZ_n|S]}{\mathbb{E}[\cZ_n]}
		&\leq&\frac{\mathbb{E}[\Zbal(\gnm)|S]}{(1+o(1))\mathbb{E}[\Zbal(\gnm)]}
		\sim\prod_{l=2}^L\left[1+\delta_l\right]^{x_l}\exp\bc{-\delta_l\lambda_l}.
	\end{eqnarray}
To obtain a matching lower bound, we claim that
	\beq\label{eqCor_count11a}
	\mathbb{E}[\cZ_n|S]\geq(1-o(1))\Erw[\Zbal(\gnm)|S]. 
	\eeq
Indeed, assume for contradiction that~(\ref{eqCor_count11a}) is false.
Then there is an $n$-independent $\eps>0$ such that for infinitely many~$n$,
	\beq\label{eqCor_count11}
	\mathbb{E}[\cZ_n|S]<(1-\eps)\Erw[\Zbal(\gnm)|S]. 
	\eeq
By Fact~\ref{Fact_cycles} there exists an $n$-independent $\xi=\xi(x_2,\ldots,x_L)>0$ such that $\pr\brk S\geq\xi$.
Hence, (\ref{eqCor_count11}) and Bayes' formula imply that
	\begin{eqnarray}
	\Erw[\cZ_n]&=&\pr\brk S\cdot \mathbb{E}[\cZ_n|S]+\pr\brk{\neg S}\mathbb{E}[\cZ_n|\neg S]\nonumber\\
		&\leq&\pr\brk S\cdot \mathbb{E}[\cZ_n|S]+\pr\brk{\neg S}\mathbb{E}[\Zbal(\gnm)|\neg S]\qquad\mbox{[as $\cZ_n\leq\Zbal(\gnm)$]}\nonumber\\
		&\leq&(1-\eps)\pr\brk S\cdot\mathbb{E}[\Zbal(\gnm)|S]+\pr\brk{\neg S}\cdot\mathbb{E}[\Zbal(\gnm)|\neg S]\nonumber\\
		&\leq&\Erw[\Zbal(\gnm)]-\eps\xi\cdot\mathbb{E}[\Zbal(\gnm)|S]\nonumber\\
		&=&\Erw[\Zbal(\gnm)]\cdot\bc{1+o(1)-\eps\xi\prod_{l=2}^L(1+\delta_l)^{x_l}\exp(-\delta_l\lambda_l)}\nonumber\\
		&=&(1-\Omega(1))\Erw[\Zbal(\gnm)]\qquad\mbox{[as $\delta_l,\lambda_l,x_l$ remain fixed as $n\ra\infty$].}
			\label{eqcor_first_moment_tame_bal11}
	\end{eqnarray}
But~(\ref{eqcor_first_moment_tame_bal11}) contradicts~(\ref{eqprop_first_moment_tame_bal1}).
Thus, we have established~(\ref{eqCor_count11a}).
Finally, combining~(\ref{eqCor_count11a}) with~(\ref{eqprop:CondRatio1stMoment1}) and~(\ref{eqprop_first_moment_tame_bal1}), we get
	\begin{eqnarray}\label{eqcor_first_moment_tame_bal12}
	\frac{\mathbb{E}[\cZ_n|S]}{\mathbb{E}[\cZ_n]}&\geq&	
		\frac{(1-o(1))\mathbb{E}[\Zbal(\gnm)|S]}{(1+o(1))\mathbb{E}[\Zbal(\gnm)]}\sim
			\prod_{l=2}^L\left[1+\delta_l\right]^{x_l}\exp\bc{-\delta_l\lambda_l},
	\end{eqnarray}
and the assertion follows from~(\ref{eqcor_first_moment_tame_bal12a}) and~(\ref{eqcor_first_moment_tame_bal12}).
\end{proof}

\noindent
We now have all the pieces in place to apply \Thm~\ref{Thm_Janson}.

\begin{corollary}\label{Cor_count}
Assume that either $k\geq3$ and $d\leq2(k-1)\ln(k-1)$ or $k\geq k_0$ for a certain constant $k_0$ and $d\leq\dc$.
Then 
	\beq\label{eqCor_count}
	\lim_{\eps\ra0}\lim_{n\ra\infty}\pr\brk{\frac{\Zkc(\gnm)}{\Erw[\Zkc(\gnm)]}\geq\eps}=1.
	\eeq
\end{corollary}
\begin{proof}
Let $\omega=\omega(n)>0$ be any sequence such that $\lim_{n\ra\infty}\omega(n)=\infty$.
Moreover, define a sequence $(\cZ_n)_{n\geq1}$ of random variables as follows.
\begin{description}
\item[Case 1: $d\leq2(k-1)\ln(k-1)$] let $\cZ_n=\Zbal(\gnm)$.
\item[Case 2: $k\geq k_0$ and $2(k-1)\ln(k-1)<d<\dc$]
	let $\cZ_n$ be equal to the random variable $\Ztame(\gnm)$ from \Prop~\ref{prop_first_moment_tame_bal}.
\end{description}
Then in either case \Prop~\ref{prop_first_moment_bal_vanilla} and~\ref{prop_first_moment_tame_bal} imply that
	\begin{eqnarray}
	\Erw[\cZ_n]&\sim&\Erw[\Zbal(\gnm)]
		\label{eqCor_count1}.
	\end{eqnarray}

We are going to apply \Thm~\ref{Thm_Janson} to the random variables $\cZ_n$ and $(C_{l,n})_{l\geq2}$.
Fact~\ref{Fact_cycles} readily implies that $C_{2,n},\ldots$ satisfy {\bf SSC1}.
Furthermore, \Prop~\ref{prop:CondRatio1stMoment} and \Cor~\ref{cor_first_moment_tame_bal} imply that for any integers $x_2,\ldots,x_L\geq0$,
	$$\frac{\mathbb{E}[\cZ_n|\forall 2\leq l\leq L:C_{l,n}=x_l]}{\mathbb{E}[\cZ_n]}\sim \prod_{l=2}^L\left[1+\delta_l\right]^{x_l}\exp\bc{-\delta_l\lambda_l}.$$
Thus, condition {\bf SSC2} is satisfied as well.
Additionally, (\ref{eqlambdadeltaConv}) establishes {\bf SSC3}.
Finally, {\bf SSC4} is verified by \Prop s~\ref{prop_second_moment_bal_vanilla} and~\ref{prop_first_moment_tame_bal}.
Hence, \Thm~\ref{Thm_Janson} applies and shows that $\cZ_n/\Erw[\cZ_n]$ converges in distribution
to $$W=\prod_{l=2}^\infty(1+\delta_l)^{X_l}\exp(-\lambda_l\delta_l),$$ where $(X_l)_{l\geq2}$ is a family of
independent random variables such that $X_l$ has distribution $\Po(\lambda_l)$.
In particular, since $W$ takes a positive (and finite) value with probability one, we conclude that
for any sequence $\omega=\omega(n)$ such that $\lim_{n\ra\infty}\omega(n)=\infty$ we have
	\begin{eqnarray}\label{eqCor_count4}
	\lim_{\delta\ra 0}\lim_{n\ra\infty}\pr\brk{\frac{\cZ_n}{\Erw[\cZ_n]}\geq\delta}&=&1.
	\end{eqnarray}
	
To complete the proof, let $(\eps(n))_{n\geq1}$ be a sequence of numbers in $(0,1)$ such that $\lim_{n\ra\infty}\eps(n)=0$.
Set $\omega(n)=-\ln\eps(n)$.
Then by \Prop~\ref{prop_first_moment_bal_vanilla} and~(\ref{eqCor_count1}) there exists an $n$-independent number $c>0$ such that
	\beq\label{eqCor_count5}
	\Erw[\Zkc(\gnm)]\leq\omega^c\cdot\Erw[\cZ_n],
	\eeq
provided that $n$ is large enough.
Thus, combining~(\ref{eqCor_count4}) and~(\ref{eqCor_count5}) and recalling that $\Zkc(\gnm)\geq\cZ_n$, we see that
	$$
	\lim_{n\ra\infty}\pr\brk{\frac{\Zkc(\gnm)}{\Erw[\Zkc(\gnm)]}\geq\eps(n)}\geq
		\lim_{n\ra\infty}\pr\brk{\frac{\cZ_n}{\Erw[\cZ_n]}\geq\omega^c\eps(n)}
		\geq\lim_{n\ra\infty}\pr\brk{\frac{\cZ_n}{\Erw[\cZ_n]}\geq\sqrt{\eps(n)}}=1.
	$$
Since this holds for any sequence $\eps(n)\ra0$, the assertion follows.
\end{proof}

\begin{proof}[Proof of \Thm~\ref{Thm_conc}]
\Cor~\ref{Cor_count} and Markov's inequality imply that
	\beq\label{eqThm_conc1}
	\lim_{\omega\ra\infty}\lim_{n\ra\infty}\pr\brk{|\ln\Zkc(\gnm)-\ln\Erw[\Zkc(\gnm)]|<\omega}=1.
	\eeq
To derive \Thm~\ref{Thm_conc} from~(\ref{eqThm_conc1}), let $S$ be the event that $\gnm$ consists of $m$ distinct edges.
Given that $S$ occurs, $\gnm$ is identical to $G(n,m)$.
Furthermore, Fact~\ref{Fact_cycles} implies that $\pr\brk{S}=\Omega(1)$.
Consequently, (\ref{eqThm_conc1}) yields
	\begin{eqnarray}\nonumber
	1&=&\lim_{\omega\ra\infty}\lim_{n\ra\infty}\pr\brk{|\ln\Zkc(\gnm)-\ln\Erw[\Zkc(\gnm)]|<\omega|S}\\
		&=&\lim_{\omega\ra\infty}\lim_{n\ra\infty}\pr\brk{|\ln\Zkc(G(n,m))-\ln\Erw[\Zkc(\gnm)]|<\omega}.
		\label{eqThm_conc2}
	\end{eqnarray}
Furthermore, (\ref{eqNaiveFirstMoment}) and \Prop~\ref{prop_first_moment_bal_vanilla} imply that
	$\Erw[\Zkc(G(n,m))],\Erw[\Zkc(\gnm)]=\Theta(k^n(1-1/k)^m)$.
Hence, $\Erw[\Zkc(\gnm)]=\Theta(\Erw[\Zkc(G(n,m))])$ and~(\ref{eqThm_conc2}) implies that
	$$\lim_{\omega\ra\infty}\lim_{n\ra\infty}\pr\brk{|\ln\Zkc(G(n,m))-\ln\Erw[\Zkc(G(n,m))]|<\omega}=1,$$
which is the first part of \Thm~\ref{Thm_conc}.

To obtain the second assertion, let $\cE_t$ be the event that the random graph $G(n,m)$ contains $t$ isolated triangles
	(i.e., $t$ connected components that are isomorphic to the complete graph on $3$ vertices).
It is well-known that for $t\geq0$ there exists $\eps=\eps(d,t)>0$ such that
	\beq\label{eqThm_conc6}
	\liminf_{n\ra\infty}\pr\brk{\cE_t}>\eps.
	\eeq
Furthermore, if given $\cE_t$ we let $G'(n,m)$ denote the random graph obtained by
choosing a set of $t$ isolated triangles randomly and removing them, then $G'(n,m)$ is identical to $G(n-3t,m-3t)$.
Hence, there is a constant $C=C(d,k)>0$ such that
	\beq\label{eqThm_conc5}
	\Erw[\Zkc(G'(n,m))]=\Erw[\Zkc(G(n-3t,m-3t))]\leq C(d,k)\cdot k^{n-3t}(1-1/k)^{m-3t}.
	\eeq
As the number of $k$-colourings of a triangle is $k(k-1)(k-2)$, (\ref{eqThm_conc5}) and~(\ref{eqNaiveFirstMoment}) yield
	\begin{eqnarray*}
	\Erw[\Zkc(G(n,m))|\cE_t]&=&\Erw[\Zkc(G(n-3t,m-3t))](k(k-1)(k-2))^t\\
		&\leq&C(d,k)\cdot k^{n}(1-1/k)^{m-3t}(1-1/k)^t(1-2/k)^t\\
		&\leq&C(d,k)\cdot k^{n}(1-1/k)^{m}\cdot (1-1/(k-1)^2)^t\\
		&\leq&O(\Erw[\Zkc(\gnm)])\cdot (1-1/(k-1)^2)^t.
	\end{eqnarray*}
Hence, for any $\omega>0$ we can choose $t$ large enough so that
	$\Erw[\Zkc(G(n,m))|\cE_t]\leq\Erw[\Zkc(\gnm)]/(2\omega)$.
In combination with Markov's inequality, this implies that
	\beq\label{eqThm_conc7}
	\pr\brk{\ln\Zkc(G(n,m))\geq\ln\Erw[\Zkc(\gnm)]-\omega|\cE_t}\leq1/2.
	\eeq
Finally, combining~(\ref{eqThm_conc6}) and~(\ref{eqThm_conc7}), 
we conclude that for any finite $\omega$ there is $\eps>0$ such that for large enough $n$,
	$$\pr\brk{\ln\Zkc(G(n,m))\geq\ln\Erw[\Zkc(\gnm)]-\omega}
		\geq\pr\brk{\ln\Zkc(G(n,m))\geq\ln\Erw[\Zkc(\gnm)]-\omega|\cE_t}\pr\brk{\cE_t}>\eps/2.$$
This completes the proof of the second claim.
\end{proof}

\begin{proof}[Proof of \Thm~\ref{Thm_cont}]
Assume for contradiction that $(\cA_n)_{n\geq1}$ is a sequence of events such that for some fixed number $0<\eps<1/2$ we have
	\beq\label{eqThm_cont0}
	\lim_{n\ra\infty}\ppl\brk{\cA_n}=0\quad\mbox{while}\quad\limsup_{n\ra\infty}\prc\brk{\cA_n}>\eps.
	\eeq
Let $G(n,m,\sigma)$ denote a graph on $[n]$ with precisely $m$ edges, such that all of these edges are bichromatic under $\sigma$, chosen uniformly at random. 
Then
	\begin{eqnarray}
	\Erw[\Zkc(G(n,m))\vecone_{\cA_n}]&=&\sum_{\sigma:\brk n\ra\brk k}\pr\brk{\sigma\mbox{ is a $k$-colouring of $G(n,m)$ and $(G(n,m),\sigma)\in\cA_n$}}\nonumber\\
		&=&\sum_{\sigma:\brk n\ra\brk k}\pr\brk{(G(n,m),\sigma)\in\cA_n|\sigma\mbox{ is a $k$-colouring of $G(n,m)$}}\nonumber\\[-4mm]
		&&		\qquad\qquad\qquad\qquad\qquad\qquad\qquad\qquad\cdot\pr\brk{\sigma\mbox{ is a $k$-colouring of $G(n,m)$}}\nonumber\\
		&=&\sum_{\sigma:\brk n\ra\brk k}\pr\brk{G(n,m,\sigma)\in\cA_n}\cdot\pr\brk{\sigma\mbox{ is a $k$-colouring of $G(n,m)$}}\nonumber\\
		&\leq&O((1-1/k)^m)\sum_{\sigma:\brk n\ra\brk k}\pr\brk{G(n,m,\sigma)\in\cA_n}\nonumber\\
		&=&O(k^n(1-1/k)^m)\pr\brk{G(n,m,\SIGMA)\in\cA_n}=o(k^n(1-1/k)^m).
			\label{eqThm_cont1}
	\end{eqnarray}

By \Cor~\ref{Cor_count}, for any $\eps>0$ there is $\delta>0$ such that for all large enough $n$ we have
	\beq\label{eqThm_cont2}
	\pr\brk{\Zkc(G(n,m))<\delta\Erw[\Zkc(G(n,m))]}<\eps/2.
	\eeq
Now, let $\cE$ be the event that $\Zkc(G(n,m))\geq\delta\Erw[\Zkc(G(n,m))]$ and let $q=\prc\brk{\cA_n|\cE}$.
Then
	\begin{eqnarray}\nonumber
	\Erw[\Zkc(G(n,m))\vecone_{\cA_n}]&\geq&\delta\Erw[\Zkc(G(n,m))]\cdot\pr\brk{((G(n,m),\TAU)\in\cA_n,\cE}\\
		&\geq&\delta q\Erw[\Zkc(G(n,m))]\pr\brk{\cE}\geq\delta q\Erw[\Zkc(G(n,m))]/2\nonumber\\
		&=&\frac{\delta q}{2}\cdot\Omega(k^n(1-1/k)^m).\label{eqThm_cont3}
	\end{eqnarray}
Combining~(\ref{eqThm_cont1}) and~(\ref{eqThm_cont3}), we obtain $q=o(1)$.
Hence, (\ref{eqThm_cont2}) implies that
	\begin{eqnarray*}
	\prc\brk{\cA_n}&=&\prc\brk{\cA_n|\neg\cE}\cdot\pr\brk{\neg\cE}+q\cdot\pr\brk{\cE}
		\leq\pr\brk{\neg\cE}+q\leq\eps/2+o(1),
	\end{eqnarray*}
in contradiction to~(\ref{eqThm_cont0}).
\end{proof}

\section{The first moment}
\label{sec_first_moment}

\noindent
The aim in this section is to prove \Prop~\ref{prop_first_moment_bal_vanilla}.
The calculations that we perform follow the path beaten in~\cite{AchNaor,Danny,WormaldColoring}.
Let $Z_{k, \rho}(G)$ be the number of $k$-colourings of the graph $G$ with colour density $\rho$.

\begin{lemma} \label{prop_first_moment_balanced} Let $k \geq 3$ and $d \in (0, \infty)$. 
	Set 
	\begin{equation}\label{eqprop_first_moment_balanced0}
	g:\rho\in\overline{\cC}_k\mapsto H(\rho) + \frac{d}{2} \ln \left( 1 - \sum_{i=1}^k \rho_i^2 \right),\quad
	\alpha(d,k) = \ln k + \frac{d}{2} \ln \left(1 - \frac{1}{k} \right), \quad c_n(d,k) =  \left({2 \pi n} \right)^{\frac{1-k}{2}} k^{k/2}.
	\end{equation}
\begin{enumerate}
\item There exist numbers $C_1=C_1(k,d), C_2=C_2(k,d) > 0$ such that
	\begin{equation}\label{eqprop_first_moment_balanced00}	
		 C_1 n^{\frac{1-k}{2}} \exp \left[ n g(\rho) \right] \leq \Erw \left[ Z_{k,\rho}(\cG(n,m)) \right] \leq C_2   \exp \left[ n g(\rho) \right]
			\quad\mbox{for any }\rho\in\cC_k(n).
		\end{equation}
	Moreover, if $\| \rho-\rho^\star\|_2 = o(1)$, then 
		\beq \label{eq_aux_first_moment_v_2} \Erw \left[ Z_{k, \rho}(\cG(n,m)) \right] \sim c_n(d,k) 
			 \exp \left[d/2+n g (\rho) \right].  \eeq
\item Assume that $\omega=\omega(n)\ra\infty$. 
Then 
	\begin{equation}\label{eqprop_first_moment_balanced000}	
	\Erw \left[\Zbal(\gnm) \right] \sim | \cB_{n,k}(\omega) | c_n(d,k) \exp \left[d/2+ n \alpha(d,k) \right].
	\end{equation}
\end{enumerate}
\end{lemma}
\begin{proof}
By Stirling's formula and the independence of the edges in the random graph $\gnm$,
	\begin{equation}\label{eqprop_first_moment_balanced11}
	\Erw[Z_{k,\rho}(\cG(n,m))]= { n \choose \rho_1 n , \dots, \rho_k n }\bc{1-\frac1N\sum_{i=1}^k\bink{\rho_in}{2}}^m,\quad\mbox{where }N=\bink n2.
	\end{equation}
Further,
	\begin{equation*} 
	\sum_{i=1}^k { \rho_{i } n \choose 2} = N \left( \sum_{i=1}^k \rho_{i }^2 \right) + \frac{n}{2} \left( \sum_{i=1}^k \rho_{i }^2-1 \right) + O(1).
	\end{equation*}
Consequently
	\begin{align} \nonumber
	m \ln \left( 1- \frac{1}{N} \sum_{i=1}^k { \rho_{i } n \choose 2} \right)& = m \ln \left[ \left( 1 + \frac{n}{2N} \right) \left(1- \sum_{i=1}^k \rho_{i }^2 \right) \right] + o(1)
	\\ & \label{eqprop_first_moment_balanced13} = n \frac{d}{2} \ln \left( 1- \sum_{i=1}^k \rho_i^2 \right) + \frac{d}{2} + o(1) .
	\end{align}

Eq. (\ref{eqprop_first_moment_balanced00}) follows from (\ref{eqprop_first_moment_balanced11}), (\ref{eqprop_first_moment_balanced13}) and Stirling's formula.
Moreover, (\ref{eq_aux_first_moment_v_2}) follows from~(\ref{eqprop_first_moment_balanced11}) 
and~(\ref{eqprop_first_moment_balanced13}) because $\norm{\rho-\rho^\star}_2=o(1)$ implies that $\sum_{i=1}^k\rho_i^2\sim1/k$ and
	\begin{equation*} 
	 { n \choose \rho_1 n , \dots, \rho_k n } \sim \left({2 \pi n} \right)^{\frac{1-k}{2}} k^{k/2} \exp \left[  n H(\rho) \right].
	 \end{equation*}
To obtain~(\ref{eqprop_first_moment_balanced000}), we observe that if $\rho \in \cB_{n,k}(\omega)$, then $\norm{\rho-\rho^\star}_2=o(1)$.
Further, by Taylor expansion we obtain
\begin{align}\label{eqprop_first_moment_balanced3}
H(\rho) &= \ln k +O \left(\sum_{i=1}^k \left(\rho_{i}-\frac{1}{k} \right)^2\right) = \ln k + o(n^{-1}),\\
 \ln \left(1 - \sum_{i=1}^k \rho_i^2\right) &= \ln \left(1 - \frac{1}{k} \right)+ O \left(\sum_{i=1}^k \left(\rho_{i}-\frac{1}{k} \right)^2\right) = \ln \left( 1 - \frac{1}{k}\right) + o(n^{-1}).  
 	\label{eqprop_first_moment_balanced4}
	\end{align}
Thus, (\ref{eqprop_first_moment_balanced000}) follows from~(\ref{eq_aux_first_moment_v_2}), (\ref{eqprop_first_moment_balanced3}) and~(\ref{eqprop_first_moment_balanced4}).
 \end{proof}
 
\begin{corollary} \label{prop_first_moment_total}
With the expressions from~(\ref{eqprop_first_moment_balanced0}),
for any $k \geq 3$ and $d \in (0, \infty)$
	$$ \Erw \left[Z_k(\gnm) \right] \sim 
		\exp \left[d/2+ n \alpha(d,k) \right]
			\left(1 + \frac{d}{k-1} \right)^{- \frac{k-1}{2}}  .$$
\end{corollary}
\begin{proof}
The functions  $\rho\in\overline{\cC}_k \mapsto H(\rho)$ and $\rho\in\overline{\cC}_k \mapsto \frac{d}{2} \ln ( 1 - \sum_{i=1}^k \rho_i^2 )$
	are both concave and attain their maximum at $\rho = \rho^\star$.
Consequently, setting
$B(d,k) = k(1 + \frac{d}{k-1})$ and expanding around $\rho = \rho^\star$, we obtain
	 \beq\label{eqVV1}
	 \alpha(d,k) - \frac{B(d,k)}{2} \| \rho - \rho^\star \|_2^2 - O\left(\| \rho - \rho^\star \|_2^3\right) \leq  g(\rho) \leq \alpha(d,k) - \frac{B(d,k)}{2} \| \rho - \rho^\star \|_2^2. 
	 \eeq
Plugging the upper bound from~(\ref{eqVV1}) into~(\ref{eqprop_first_moment_balanced00}) and observing that $|\cC_{n,k}|\leq n^{k}=\exp(o(n))$, we find
   \beq\label{eqVV2}
   S_1=\sum_{\substack{ \rho \in \cC_{n,k} \\\| \rho - \rho^\star \|_2 > n^{-5/12}} } 
  	\Erw \left[ Z_{k,\rho}(\cG(n,m)) \right] \leq C_2   \exp \left[ \alpha(d,k) \right] \exp \left[ - \frac{B(d,k)}{2} n^{1/6} \right] . 
		 \eeq
On the other hand, (\ref{eq_aux_first_moment_v_2}) implies that
 \begin{align} \nonumber 
	S_2&=\sum_{\substack{ \rho \in \cC_{n,k} \\\| \rho - \rho^\star \|_2 \leq n^{-5/12}} } \Erw \left[ Z_{k,\rho}(\cG(n,m)) \right] 
		 \sim  \sum_{\substack{ \rho \in \cC_{n,k} \\\| \rho - \rho^\star \|_2 \leq n^{-5/12}} } c_n(d,k) \exp(d/2)  \exp \left[ n g(\rho) \right]
		 \\ & \sim c_n(d,k) 
		 	 \exp\left[d/2+ n \alpha(d,k) \right] 
		 	\sum_{\substack{ \rho \in \cC_{k}(n) } } \exp \left[ - n \frac{B(d,k)}{2} \| \rho - \rho^\star \|_2^2\right]  . \label{eq_aux_first_moment_v_1} 
 	\end{align}
The last sum is almost in the standard form of a Gaussian summation, just that
the vectors $\rho\in\cC_k(n)$ that we sum over are subject to the linear constraint $\rho_1+\cdots+\rho_k=1$.
We rid ourselves of this constraint by substituting $\rho_k=1-\rho_1-\cdots-\rho_{k-1}$.
Formally, let
	$J$ be the $(k-1)\times(k-1)$-matrix whose diagonal entries are equal to $2$ and whose remaining entries are $1$.
Then 
	 \begin{align}  \nonumber
	 \sum_{\substack{ \rho \in \cC_{n,k} } } \exp \left[ - n \frac{B(d,k)}{2} \| \rho - \rho^\star \|_2^2\right]   &\sim 
	 	\sum_{y \in \frac1n\mathbb{Z}^k} \exp \left[ - n \frac{B(d,k)}{2}\scal{Jy}y \right]\\
		&\sim\left( 2 \pi n \right)^{\frac{k-1}{2}} k^{- \frac{k}{2}} \left(1 + \frac{d}{k-1} \right)^{- \frac{k-1}{2}}
		&[\mbox{as $\det J=k$}].
			\label{eq_aux_first_moment_v_666} 
	 \end{align}
Plugging~(\ref{eq_aux_first_moment_v_666}) into (\ref{eq_aux_first_moment_v_1}), we obtain
    \begin{align}\nonumber
    S_2
    	&\sim c_n(d,k) \exp \left[d/2+ n \alpha(d,k) \right] \left( 2 \pi n \right)^{\frac{k-1}{2}} k^{- \frac{k}{2}} \left(1 + \frac{d}{k-1} \right)^{- \frac{k-1}{2}}\\
	&=\exp \left[d/2+ n \alpha(d,k) \right]\left(1 + \frac{d}{k-1} \right)^{- \frac{k-1}{2}}&\mbox{[using~(\ref{eqprop_first_moment_balanced0})]}.
				\label{eq_aux_first_moment_v_667} 
  \end{align}
Finally, comparing~(\ref{eqVV2}) and~(\ref{eq_aux_first_moment_v_667}), 
we see that $S_1=o(S_2)$.
Thus, $\Erw [ Z_k(\cG(n,m))] =S_1+S_2\sim S_2$, and the assertion follows from~(\ref{eq_aux_first_moment_v_667}).
\end{proof}

\begin{proof}[Proof of \Prop~\ref{prop_first_moment_bal_vanilla}]
The first assertion is immediate from \Cor~\ref{prop_first_moment_total}.
Moreover, the second assertion follows from \Cor~\ref{prop_first_moment_total} and the second part of \Lem~\ref{prop_first_moment_balanced}.
\end{proof}

\section{Counting short cycles}\label{sec:prop:CondRatio1stMoment}

\noindent
Throughout this section, we let $x_2,\ldots,x_L$ denote a sequence of non-negative integers.
Moreover, let $S$ be the event that $C_{l,n}=x_l$ for $l=2,\ldots,L$.
Additionally,  let $\cV(\sigma)$ be the event that $\sigma$ is a $k$-colouring of the random graph $\gnm$.
We also recall $\lambda_l,\delta_l$ from~(\ref{eqlambdadelta}).

\subsection{Proof of \Prop~\ref{prop:CondRatio1stMoment}} 
The key ingredient to the proof is the following lemma
concerning the distribution of the random variables $C_{l,n}$ given $\cV(\sigma)$.

\begin{lemma}\label{prop:planted-cyclces}
Let $\mu_l=\frac{d^l}{2l}\left[1+\frac{(-1)^l}{(k-1)^{l-1}}\right]$.
Then
$\pr[S|\cV(\sigma)] \sim \prod_{l=2}^{L}\frac{\exp\bc{-\mu_l}}{x_l!}\mu_l^{x_l}$ for any $\sigma\in\Bal$.
\end{lemma}
\noindent

\noindent
Before we establish \Lem~\ref{prop:planted-cyclces},
let us point out how it implies \Prop~\ref{prop:CondRatio1stMoment}. 
By Bayes' rule,
\begin{eqnarray}
\Erw\left[\Zbal(\gnm)| S\right]&=& \frac{1}{\pr[S]}\sum_{\tau\in\Bal}\pr[\cV(\tau)] \pr[S|\cV(\tau)] \qquad \nonumber \\
&\sim& \frac{
	\prod_{l=2}^{L}\frac{\exp\bc{-\mu_l}}{x_l!}\mu_l^{x_l}}{\pr[S]}
		\sum_{\tau\in [k]^n:\tau\in B}\pr[\cV(\tau)] \qquad \mbox{[from \Lem~\ref{prop:planted-cyclces}]} \nonumber\\
&\sim& \frac{\prod_{l=2}^{L}\frac{\exp\bc{-\mu_l}}{x_l!}\mu_l^{x_l}}{\pr[S]}\mathbb{E}[\Zbal(\gnm)].   \nonumber 
\end{eqnarray}
From \Lem~\ref{prop:planted-cyclces} and Fact~\ref{Fact_cycles} we get that
\begin{eqnarray}
\frac{\prod_{l=2}^{L}\frac{\exp\bc{-\mu_l}}{x_l!}\mu_l^{x_l}}{\pr[S]}\sim\prod_{l=2}^L\left[1+\delta_l\right]^{x_l}\exp\bc{-\delta_l\lambda_l}, \nonumber
\end{eqnarray}
whence \Prop~\ref{prop:CondRatio1stMoment} follows. \hfill $\square$

\subsection{Proof of \Lem~\ref{prop:planted-cyclces}}\label{sec:prop:planted-cyclces}

We are going to show that
for any fixed sequence of integers $m_2, \ldots, m_L\geq0$, the joint factorial moments
satisfy 
\begin{eqnarray}
\mathbb{E}\left[(C_{2,n})_{m_2}\cdots (C_{L,n})_{m_L}|\cV(\sigma)\right] \sim
\prod_{l=2}^L \mu_l^{m_l}. \label{eq:TargetRelatJointFactMoment}
\end{eqnarray}
Then \Lem~\ref{prop:planted-cyclces} follows from~\cite[\Thm~1.23]{BB}.

We  consider the number of sequences of $m_2+\cdots+m_L$ distinct cycles such that 
 $m_2$ corresponds to the number of cycles of length $2$, and so on.  Clearly this number is equal
to $(C_{2,n})_{m_2}\cdots (C_{L,n})_{m_L}$.  
Let $Y$ be the number of those sequences of cycles such that any two cycles are vertex-disjoint.
Also, let $Y'$ denote the number of sequences which have intersecting
cycles.  Clearly it holds that 
\begin{eqnarray}
\mathbb{E}\left[(C_{2,n})_{m_2}\cdots (C_{L,n})_{m_L}| \cV(\sigma)\right]=\mathbb{E}[Y|\cV(\sigma)]+\mathbb{E}[Y'|\cV(\sigma)].\label{eq:SplitToDisjointJointCase}
\end{eqnarray}
For $\mathbb{E}[Y'|\cV(\sigma)]$ we use the following claim, whose proof follows below.

\begin{claim}\label{claim:OverlapCycles}
It holds that $\mathbb{E}[Y'|\cV(\sigma)]=O(n^{-1})$.
\end{claim}
\noindent

Hence, we need to count vertex disjoint cycles given $\cV(\sigma)$.
To this end, we adapt the argument  for random regular graphs from~\cite[\Sec~2]{WormaldColoring}.
Thus, we consider rooted, directed cycles, first. This will introduce a factor of $2l$ for the number of cycles
of length $l$. That is, if $D_l$ is the number of rooted, directed cycles of length $l$ then $D_l=2lC_l$.

For a rooted directed cycle $(v_1,\ldots,v_l)$ of length $l$, we call $(\sigma(v_1),\ldots,\sigma(v_l))$ the {\em type} of the cycle under $\sigma$.
For $t=(a_1,\ldots,a_l)$ let
 $D_{l,t}$ denote the number of rooted, directed cycles (of length $l$ and) type $t$.  
We claim that
\begin{align}\label{eqCharisCycles1}
\mathbb{E}\left[D_{l,t} |\cV(\sigma)\right]
	\sim\bcfr nk^l\frac{(m)_l}{N^l (1-\Forb(\sigma)/N)^l}\sim\bcfr{d}{k-1}^l\quad\mbox{with }N=\bink n2.
\end{align}
Indeed, since $\sigma$ is $(\omega,n)$-balanced, the number of ways of choosing a vertex of colour $t_i$ is $(1+o(1))n/k$,
and we have got to choose $l$ vertices in total.
Thus, the total number of ways of choosing $l$ vertices $(v_1,\ldots,v_l)$ such that $\sigma(v_i)=t_i$ for all $i$ is $(1+o(1))(n/k)^l$.
In addition, each edge $\cbc{v_i,v_{i+1}}$ of the cycle is present in the graph with a probability asymptotically equal to $m / (N-\Forb(\sigma))$
This explains the first asymptotic equality in~(\ref{eqCharisCycles1}).
The second one follows because $m\sim dn/2$ and $\Forb(\sigma)\sim 1/kN$ (as $\sigma\in\Bal$).

In particular, the r.h.s.\ of~(\ref{eqCharisCycles1}) is independent of the type $t$.
For a given $l$ let $T_l$ signify the number of all possible types of cycles of length $l$.
Thus, $T_l$ is the set of all sequences $(t_1,\ldots,t_l)$ such that $t_{i+1}\neq t_i$ for all $1\leq i<l$ and $t_l\neq t_1$.
Let $T_1=0$.
Then $T_l$ satisfies the recurrence $T_l+T_{l-1}=k(k-1)^{l-1}$
	(cf.~\cite[\Sec~2]{WormaldColoring}).\footnote{To see this, observe that $k(k-1)^l$ is the number of all sequences
	$(t_1,\ldots,t_l)$ such that $t_{i+1}\neq t_i$ for all $1\leq i<l$.
	Any such sequence either satisfies $t_l\neq t_1$, which is accounted for by $T_l$, or $t_l=t_1$ and $t_{l-1}\neq t_1$,
	in which case it is contained in $T_{l-1}$.}
Hence, $T_l=(k-1)^l+(-1)^l(k-1)$.
Combining this formula with~(\ref{eqCharisCycles1}), we obtain
\begin{eqnarray}
\mathbb{E}\left[D_l |\cV(\sigma) \right]&\sim& T_l \cdot \mathbb{E}\left[D_{l,t} |\cV(\sigma) \right] 
\sim\left(1+\frac{(-1)^l}{(k-1)^{l-1}}\right) \cdot d^{l}. \nonumber 
\end{eqnarray}
Hence, recalling that $C_l=\frac1{2l}D_l$, we get
\begin{eqnarray}
\mathbb{E}\left[C_l |\cV(\sigma)\right]&\sim&  \frac{d^l}{2l}\left[1+\frac{(-1)^l}{(k-1)^{l-1}}\right]. \label{eqCharis2}
\end{eqnarray}
In fact, since $Y$ considers only vertex disjoint cycles  and $l$, $m_2,\ldots, m_L$ remain fixed as $n\ra\infty$, (\ref{eqCharis2}) yields
\begin{eqnarray}
\mathbb{E}[Y|\cV(\sigma)] \sim\prod_{l=2}^L \left(\frac{d^l}{2l}\left[1+\frac{(-1)^l}{(k-1)^{l-1}}\right]\right)^{m_l} \nonumber.
\end{eqnarray}
Plugging the above relation and Claim  \ref{claim:OverlapCycles} into (\ref{eq:SplitToDisjointJointCase}) we get 
(\ref{eq:TargetRelatJointFactMoment}). The proposition follows. \hfill $\square$
\\

\begin{claimproof}{ \ref{claim:OverlapCycles}}
For every subset $R$ of $l$ vertices, where $l\leq L$  let  $\mathbb{I}_{R}$ be equal to 1 if the number of edges with 
both end in $R$ is at least $|R|+1$. 
Let the event $H_L=\{\sum_{R:|R|\leq L}\mathbb{I}_{R}>0\}$.  It is direct to check that if $Y'>0$   then 
the event $H_L$ occurs. This implies that
\begin{eqnarray}
\pr[Y'>0|\cV(\sigma)]\leq \pr[H_L|\cV(\sigma)]. \nonumber
\end{eqnarray}
The claim follows by bounding appropriately $\pr[H_L|\cV(\sigma)]$.
For this we are going to use Markov's inequality, i.e. 
\begin{eqnarray}
\pr[H_L|\cV(\sigma)]&\leq& \mathbb{E}\left[\sum_{R:|R|\leq L}\mathbb{I}_{R}|\cV(\sigma) \right] \nonumber
	=\sum_{l=1}^L\sum_{R:|R|=l}\mathbb{E}\left[\mathbb{I}_{R}|\cV(\sigma) \right] \nonumber.
\end{eqnarray}
For any set $R$ such that $|R|=l$, 
we can put $l+1$ edges inside  the set in at most ${{l\choose 2}\choose l+1}$ ways.
Clearly conditioning on $\cV(\sigma)$ can only reduce the  number of different 
placings of the edges.

Using inclusion/exclusion, for a fixed set $R$ of cardinality $l$ we get that
\begin{eqnarray}
\mathbb{E}\left[\mathbb{I}_{R}|\cV(\sigma)\right]&\leq& {{l\choose 2}\choose l+1 }    
{\sum_{i=0}^{l+1}{l+1 \choose i}(-1)^i\left(1-\frac{i}{N-\cF\bc{\sigma}}\right)^m}
\nonumber\\
&\leq& {{l\choose 2}\choose l+1} \left(\frac{m}{N-\cF\bc{\sigma}}\right)^{l+1} \hspace*{2cm}
\mbox{[from the Binomial theorem]} \nonumber\\
&\sim& {{l\choose 2}\choose l+1} \left(\frac{d}{n(1-1/k)}\right)^{l+1}. \hspace*{2cm} \mbox{[since $m=\frac{dn}{2}$, and  $\cF\bc{\sigma}\sim \frac{1}{k}N$].}
\nonumber
\end{eqnarray}
It holds that
\begin{eqnarray}
\pr[H_m|\cV(\sigma)] &\leq& (1+o(1))\sum_{l=1}^L{n \choose l}{{l\choose 2}\choose l+1} \left(\frac{d}{n(1-1/k)}\right)^{l+1} \nonumber\\
&\leq&(1+o(1)) \sum_{l=1}^L\left(\frac{ne}{l}\right)^l\left(\frac{le}{2}\right)^{l+1} \left(\frac{d}{n(1-1/k)}\right)^{l+1} 
 \qquad\mbox{[since ${i\choose j}\leq \left(ie/j\right)^j$]}\nonumber\\
&\leq& \frac{1+o(1)}{n} \sum_{l=1}^L\frac{led}{2(1-1/k)}\left(\frac{e^2d}{2(1-1/k)}\right)^l  =O(n^{-1}),\nonumber
\end{eqnarray}
the last equality holds since $L$ is a fixed number.  The claim  follows.
\end{claimproof}

\section{The second moment computation} \label{sec_second_moment}

\noindent
In this section we prove the second moment bounds claimed in \Prop s~\ref{prop_second_moment_bal_vanilla} and~\ref{prop_first_moment_tame_bal}, which
constitute  the main technical contribution of this work.
While here we need an asymptotically tight expression for the second moment,
in prior work on colouring $G(n,m)$ the second moment was merely computed {\em up to a constant factor}~\cite{AchNaor,Cond,Danny}.
Only in the case of random regular graphs was the second moment computed up to a factor of $1+o(1)$~\cite{WormaldColoring}.
In addition, all of these papers confine themselves to the case of colourings whose colour densities are $(O(1),n)$-balanced,
whereas here we need to deal with $(\omega,n)$-balanced colour densities for a diverging function $\omega=\omega(n)\ra\infty$.

Thus, the plan is to extend the arguments from~\cite{AchNaor,Cond,Danny} to get a precise asymptotic result,
and to cover the $(\omega,n)$-balanced case.
Unsurprisingly, in the course of this we will frequently encounter formulas that resemble those of~\cite{AchNaor,Cond,Danny},
and occasionally we will be able to reuse some of the calculations done in those papers.
Furthermore, to determine the precise constant we can harness a bit of linear algebra from~\cite{WormaldColoring}.
Throughout this section $\omega=\omega(n)$ stands for a function that tends to $\infty$ (slowly).

\subsection{The overlap}
Following~\cite{AchNaor}, 
for $\sigma, \tau : [n] \to [k]$ we define the \textit{overlap matrix} $\rho(\sigma, \tau)=(\rho_{ij}(\sigma, \tau))_{i,j\in[k]}$ as the $k\times k$-matrix with
entries
	$$ \rho_{ij}(\sigma, \tau) = \frac{1}{n} \cdot | \sigma^{-1}(i) \cap \tau^{-1}(j) |. $$
Moreover, for a $k\times k$-matrix $\rho=(\rho_{ij})$ we introduce the shorthands
	$$ \rho_{i \star} = \sum_{j=1}^k \rho_{ij}, \qquad\rho_{\nix\star}=(\rho_{i \star})_{i\in[k]},
	 \qquad \qquad \rho_{\star j} = \sum_{i=1}^k \rho_{ij},\qquad
	 	\rho_{\star\nix}=(\rho_{\star i  })_{i\in[k]}. $$
Thus, for any $\sigma, \tau : [n] \to [k]$ we have
	 $\rho_{\nix\star},\rho_{\star\nix}\in\cC_{k}(n)$.

Let $\overline{\cR}_k$ denote the set of all probability measures $\rho=(\rho_{ij})_{i,j\in[k]}$ on $[k]\times[k]$ and let
$\bar\rho$ signify the $k\times k$-matrix with all entries equal to $k^{-2}$, the barycentre of $\overline{\cR}_k$.
Additionally, we introduce 
	\begin{align*}
	\cR_{n,k}&=\cbc{\rho(\sigma,\tau):\sigma,\tau:[n]\ra[k]},\\
	\cR_{n,k}^{\rm int}&=\cbc{\rho\in\cR_{n,k}:\mbox{$\rho_{ij} > 1/{k^3}$ for all $i,j \in [k]$}},\\
	\cR_{n,k}^{\bal}(\w)&=\cbc{\rho\in\cR_{n,k}^{\rm int}:
		| \rho_{i \star} - k^{-1}| \leq \w^{-1} n^{-1/2},| \rho_{\star i} - k^{-1}| \leq \w^{-1} n^{-1/2}\mbox{ for all $i\in [k]$}},\\
	\cR_{n,k}^{\bal}(\w,\eta)&=\cbc{\rho\in\cR_{n,k}^{\bal}(\w):\norm{\rho-\bar\rho}_2\leq\eta}\qquad\mbox(\eta>0).
	\end{align*}

For a given graph $G$ on $[n]$, let $Z^{(2)}_{k, \rho}(G)$ be the number of pairs $(\sigma,\tau)$ of $k$-colourings  of $G$ whose overlap 
is $\rho$.
Then by the linearity of expectation,
	\begin{align}\label{eqsmmAmin}
	\Erw \left[\Zbal(\gnm)^2 \right]&=
		\sum_{\rho\in\cR_{n,k}^{\bal}(\w)}\Erw[Z^{(2)}_{k, \rho}(\gnm)].
	\end{align}
We are going to show that the r.h.s.\ of~(\ref{eqsmmAmin}) is dominated by the contributions with $\rho$ ``close to'' $\bar\rho$.
More precisely, let
	$$Z^{(2)}_{k,\w,\eta}(G)=\sum_{\rho\in\cR_{n,k}^{\bal}(\w,\eta)}Z^{(2)}_{k, \rho}(G)\qquad
		\mbox{for any }\eta>0.$$
Then the second moment argument performed in~\cite{AchNaor} fairly directly yields the following statement.

\begin{proposition}
\label{prop_ach_naor}
Assume that $k \geq 3$ and that $d < 2(k-1) \ln (k-1)$. Then for any fixed $\eta >0$ it holds that $$\Erw [\Zbal( \cG(n,m))^2] \sim \Erw [Z^{(2)}_{k, \w, \eta}(\cG(n,m))].$$
\end{proposition}

\noindent
In addition, the second moment argument from~\cite{Danny} implies

\begin{proposition}
\label{prop_aco_vil}
There is a constant $k_0>3$ such that for
 $k\geq k_0$ and that $2(k-1) \ln (k-1) \leq d < \dc$ the following is true.
There exists an integer-valued random variable
	$0 \leq \Ztame \leq \Zbal$ that satisfies $$ \Erw[ \Ztame(\gnm) ] \sim \Erw \left[ \Zbal(\gnm) \right]$$ 
and such that for any fixed $\eta >0$ we have $\Erw [\Ztame(\cG(n,m))^2] \leq(1+o(1)) \Erw [Z^{(2)}_{k, \w, \eta}(\cG(n,m))].$
\end{proposition}

\noindent
Since the above statements do not quite appear in this form in~\cite{AchNaor,Danny},
we will prove them in \Sec s~\ref{Sec_prop_ach_naor} and~\ref{Sec_prop_aco_vil}, respectively.

\subsection{Homing in on $\bar\rho$}
Having reduced our task to 
studying overlaps $\rho$ such that 
$\norm{\rho-\bar\rho}_2\leq\eta$ for a small but fixed $\eta>0$,
in this section we are going to argue that, in fact, it suffices to consider $\rho$ such that
	$\norm{\rho-\bar\rho}_2\leq n^{-5/12}$
(where the constant $5/12$ is somewhat arbitrary; any number smaller than $1/2$ would do).
More precisely, we have

\begin{proposition} \label{prop_second_moment_new_v}
Assume that $k \geq 3$ and that $d < \dc$. There exists a number $\eta_0=\eta_0(d,k)$
such that for any $0<\eta<\eta_0$ we have %
$$ \Erw [Z^{(2)}_{k, \w, \eta}(\cG(n,m))] \sim \Erw [Z^{(2)}_{k, \w, n^{-5/12}}(\cG(n,m))]. $$
\end{proposition}

\noindent
In order to prove \Prop~\ref{prop_second_moment_new_v}, we first need the following elementary estimates.
\begin{fact}\label{fact_Z_rho_int}\label{fact_Z_rho_bal}
For any $k \geq 3$, $d \in (0, \infty)$ the following estimates are true.
\begin{enumerate}
\item Let  $\rho \in \cR_{n,k}^{\rm int}$. %
	Then
	\beq \begin{split} \label{eq_Z_rho_int}
	\Erw  \left[ Z^{(2)}_{k,\rho}(\cG(n,m)) \right]  \sim &
		\frac{\sqrt{2 \pi} n^{\frac{1-k^2}{2}}}
			{\prod_{i,j=1}^k \sqrt{2 \pi \rho_{ij}}} \exp [d/2+ n  H(\rho) 
				+m \ln( 1 - \normA^2-\normB^2+\normC^2)] 
	\end{split} \eeq 
\item	For any $\rho \in \mathcal{R}_{n,k}^{\rm bal}(\w) $ we have
	\begin{equation} \label{eq_Z_2_balanced_colorings} 
	\Erw  \left[ Z^{(2)}_{k,\rho}(\cG(n,m)) \right]  \sim 
		\frac{\sqrt{2 \pi} n^{\frac{1-k^2}{2}}}
			{\prod_{i,j=1}^k \sqrt{2 \pi \rho_{ij}}} \exp [d/2+ n  H(\rho) 
				+m 
	\ln(1 - 2/k +\normC^2 )].
			 \end{equation}
\end{enumerate}
\end{fact}
\begin{proof} By Stirling's formula, the total number of $\sigma, \tau$ with overlap $\rho \in \cR_{n,k}^{\rm int}$ is given by:
\beq \label{eq_aux_second_v_1} { n \choose \rho_{11} n, \dots, \rho_{kk} n} \sim \sqrt{2 \pi} n^{- \frac{k^2-1}{2}} \left( \prod_{i,j} \frac{1}{\sqrt{2 \pi \rho_{ij}}} \right) \exp \left[ n  H(\rho) \right]. \eeq
To obtain $\Erw  \left[ Z^{(2)}_{k,\rho}(\cG(n,m)) \right] $, we need to multiply this number by the probability that two maps $\sigma, \tau$ with overlap $\rho$ are both colourings of a randomly chosen graph. The number of ``forbidden'' edges joining two vertices with the same colour under either $\sigma$ or $\tau$ is given by
\begin{equation*} \begin{split} \cF(\sigma, \tau) &= \sum_{i=1}^k { \rho_{i \star} n \choose 2}+ \sum_{j=1}^k { \rho_{\star j} n \choose 2} - \sum_{i,j=1}^k { \rho_{ij} n \choose 2} 
\\ & = N \left( \sum_{i=1}^k \rho_{i \star}^2 + \sum_{j=1}^k \rho_{\star j}^2 - \sum_{i,j = 1}^k \rho_{ij}^2  \right) + \frac{n}{2} \left( \sum_{i=1}^k \rho_{i \star}^2+\sum_{j=1}^k \rho_{\star j}^2 - \sum_{i,j=1}^k \rho_{ij}^2-1 \right) + O(1)
. \end{split} \end{equation*}
Therefore, the probability that $\sigma$ and $\tau$ are both colourings of $\gnm$ depends only on their overlap $\rho$, and is
\begin{align} \nonumber \mathbb{P} \left[ \sigma, \tau \textrm{ are $k$-colourings of $\gnm$} \right] &=  \frac{\left(N - \mathcal{F}(\sigma, \tau)\right)^m}{N^m} 
\\ \label{eq_aux_second_v_2} & \sim  \exp \left[ m \ln \left( 1 - \sum_{i=1}^k  \rho_{i \star}^2- \sum_{j=1}^k  \rho_{\star j}^2  + \sum_{i,j=1}^k \rho_{ij}^2 \right) + \frac{d}{2} \right].
  \end{align}
Eq. (\ref{eq_Z_rho_int}) is obtained by multiplying (\ref{eq_aux_second_v_2}) with (\ref{eq_aux_second_v_1}).

To prove the second claim,
let $\epsilon _i = \rho_{i \star} - 1/k$ for $i\in[k]$.
Because $\sum_{i,j=1}^k \rho_{ij} = 1$ we have $\sum_{i=1}^k \epsilon_i = 0$.
Consequently,
	\beq\label{eqVV99}
	\normA^2=\frac1k+\sum_{i=1}^k \epsilon_i^2.
	\eeq
Further, if $\rho$ is $(\w,n)$-balanced, then $\epsilon_i =o(n^{-1/2})$ for all $i \in [k]$.
Hence, (\ref{eqVV99}) yields $\normA^2 = \frac{1}{k} + o(n^{-1})$.
Similarly,  $\normB^2  
	= \frac{1}{k} + o(n^{-1})$.
Therefore, for any $(\w,n)$-balanced $\rho$,
$$   m \ln \left( 1 - \normA^2-\normB^2+\normC^2 
	\right)
	 \sim m \ln \left( 1 - \frac{2}{k} + \normC^2 
	 	 \right) .$$
Plugging the above into (\ref{eq_Z_rho_int}) completes the proof.
  \end{proof}

To evaluate the exponential part in Eq. (\ref{eq_Z_2_balanced_colorings}), we require the following \Lem.

\begin{lemma}\label{lemma_expansion_saddle}
Let $k \geq 3$ and $d < (k-1)^2 $.
Let $\alpha(d,k)$ be as in (\ref{prop_first_moment_balanced}) and set
	$$ C_n(d,k) = \exp(d/2) k^{k^2} \left(2 \pi n\right)^{\frac{1-k^2}{2}} , \qquad D(d,k) = k^2 \left(1 - \frac{d}{\left(k-1 \right)^2} \right).$$
\begin{itemize}
\item If $\rho \in \cR_{n,k}^\bal(\w)$ satisfies $\| \rho - \bar{\rho} \|_2 \leq n^{-5/12}$, then 
	\beq  \label{eq_aux_bound_Z_rho_1}
		\Erw \left [Z_{k,\rho}^{(2)} (\cG(n,m)) \right] \sim C_n(d,k) \exp \left[ 2 n \alpha(d,k) -n  \frac{D(d,k)}{2}  \| \rho - \bar{\rho}\|_2^2 \right] .\eeq
\item There exist numbers $\eta=\eta(d,k)  >0$ and $A=A(d,k) >0$ such that if $ \rho \in \cR_{n,k}^\bal(\w)$ satisfies $\| \rho - \bar{\rho} \|_2 \in (n^{-5/12}, \eta)$, then
	\beq \label{eq_aux_bound_Z_rho_2}
		\Erw \left[ Z_{k,\rho}^{(2)} (\cG(n,m)) \right] = %
			 \exp \left[2 n \alpha(d,k) %
			 	 - A n^{1/6} \right] . \eeq
\end{itemize}
\end{lemma}
\begin{proof} 
Following~\cite{AchNaor}, we consider
\beq \label{eq_def_f_rho} 
f : \overline{\cR}_k \to \mathbb{R},\quad
\rho\mapsto H(\rho) + \frac{d}{2} \ln \left( 1 - \frac{2}{k} + \sum_{i,j=1}^k \rho_{ij}^2 \right) . \eeq
Then Fact~\ref{fact_Z_rho_bal} yields
	$ \Erw [ Z_{k,\rho}^{(2)} (\cG(n,m)) ] \sim C_n(d,k) \exp \left[ n f(\rho) \right] $.
The function $f$ satisfies $f(\bar{\rho}) = 2 \alpha(d,k)$.
Further, expanding $f$ around $\bar{\rho}$ by writing $\epsilon = \rho - \bar{\rho}$ (so that $\sum_{i,j=1}^k \epsilon_{ij} = 0$) gives
\begin{align} \nonumber \label{eq_aux_expansion_f_bal} f(\rho)& =  H(\bar{\rho}) - \frac{k^2}{2} \sum_{i,j=1}^k \epsilon_{ij}^2 + O \left( \|\epsilon \|^3_2 \right) + \frac{d}{2} \ln \left( 1- \frac{2}{k} + \frac{1}{k^2} + \sum_{i,j=1}^k \epsilon_{ij}^2  \right)
\\ & = f(\bar{\rho}) - \frac{D(d,k)}{2} \| \epsilon \|_2^2 + O(\|\epsilon\|_2^3).  \end{align}
Consequently for $\| \rho - \bar{\rho} \|_2 \leq n^{-5/12}$,
	$$\exp \left[ n f(\rho) \right] = \exp \left[ n f(\bar{\rho}) - n \frac{D(d,k)}{2} \| \rho- \bar{\rho}\|_2^2  + O(n^{-1/4}) \right],$$ whence~(\ref{eq_aux_bound_Z_rho_1}) follows.

We now prove Eq. (\ref{eq_aux_bound_Z_rho_2}). Similarly to (\ref{eq_aux_expansion_f_bal}) and because $f$ is smooth in a neighborhood of $\bar{\rho}$, 
there exist $\eta >0$ and $A >0$ such that for $\| \rho - \bar{\rho} \|_2 \leq \eta$, 
$$ f(\rho) \leq f(\bar{\rho}) - A\| \rho - \bar{\rho} \|_2^2 . $$
Hence, if $\| \rho - \bar{\rho} \|_2 \in (n^{-5/12}, \eta)$, then
		 $$\Erw \left[Z_{k,\rho}^{(2)}(\cG(n,m))\right] = O \left(n^{\frac{1-k^2}{2}} \right)
		 	\exp \left[ n f(\rho) \right] \leq %
					\exp \left[2 n \alpha(d,k) %
						 - A n^{1/6} \right],$$ 
as claimed.
 \end{proof}

\begin{proof}[Proof of \Prop~\ref{prop_second_moment_new_v}] We fix $\eta >0$ and $A>0$ as given by \Lem~\ref{lemma_expansion_saddle}. 
Fixing $\rho_0 \in \cR_{n,k}^\bal(\w,\eta)$ such that $\| \rho_0- \bar{\rho}\|_2 \leq k/n$, we obtain from the first part of \Lem~\ref{lemma_expansion_saddle} that
	\beq\label{eqVV01}
	\Erw [Z^{(2)}_{k, \w, n^{-5/12}}(\cG(n,m))]  \geq \Erw \left[Z^{(2)}_{k,\rho_0} (\gnm) \right] \sim C_n(d,k) \exp \left[ 2 n \alpha(d,k) \right].
	\eeq
On the other hand, 
because $|\cR_{n,k}^\bal(\omega, \eta)|$ is bounded by a polynomial in $n$,
the second part of \Lem~\ref{lemma_expansion_saddle} yields
\begin{align} \label{eqVV02}
\sum_{\substack{\rho \in \cR_{n,k}^\bal(\w,\eta) \\  \| \rho - \bar{\rho} \|_2 > n^{-5/12} }} \Erw \left[ Z_{k,\rho}^{(2)}(\cG(n,m)) \right] &
 \leq  \exp \left[ 2 n \alpha(d,k)-A n^{1/6} + O(\ln n) \right] . %
  \end{align} 
Combining~(\ref{eqVV01}) and~(\ref{eqVV02}), we obtain
 \begin{align*} \Erw [Z^{(2)}_{k, \w, \eta}(\cG(n,m))]  &\sim  \sum_{\substack{\rho \in \cR_{n,k}^\bal(\w,n^{-5/12})}} \Erw \left[ Z_{k,\rho}^{(2)}(\cG(n,m)) \right] \sim 
 \Erw [Z^{(2)}_{k, \w, n^{-5/12}}(\cG(n,m))], \end{align*}
as claimed.
 \end{proof}

\subsection{The leading constant.}
Here we compute the contribution of overlap matrices $\rho \in \cR_{n,k}^{\bal}(\w,n^{-5/12})$. 

\begin{proposition}\label{prop_second_moment_balanced_exact} Assume that
 $k \geq 3$, $d < (k-1)^2$. Then  with $c_n(d,k)$ from~(\ref{eqprop_first_moment_balanced0}),
	$$ \Erw \left[ Z^{(2)}_{k,\w,n^{-5/12}}(\cG(n,m)) \right]
		 \sim \left( | \cB_{n,k}(\omega) | c_n(d, k) \exp \left[ n \alpha(d,k) \right] \right)^2 \exp(d/2) \left( 1- \frac{d}{(k-1)^2} \right)^{-\frac{(k-1)^2}{2} }.   $$ \end{proposition}
In order to prove the \Prop, we will need the following lemma regarding Gaussian summations over matrices with coefficients in
$\frac1n\mathbb{Z}$ whose lines and columns sums to zero.
Thus, let
	\begin{equation} \label{eq_def_cSn} \cS_n = \cbc{ \left(\epsilon_{i,j}\right)_{\substack{1 \leq i \leq k \\ 1 \leq j \leq k} },\;\; \forall i, j \in [k],\; \epsilon_{i,j} \in \frac1n\mathbb{Z},\;
			\; \forall j \in [k],\;  \sum_{i = 1}^k \epsilon_{i j} = \sum_{i=1}^k \epsilon_{j i} = 0  } .
	\end{equation}

\begin{lemma}\label{lemma_hessian_wormald}
Let $k \geq 2$, $d < (k-1)^2$ and $D >0$ be fixed. Then
\beq \label{eq_aux_v_5} \sum_{\epsilon \in S_n} \exp \left[ - n \frac{D}{2} \| \epsilon\|_2^2 + o(n^{1/2}) \| \epsilon\|_2 \right] \sim \left( \sqrt{ {2 \pi}{n} }\right)^{(k-1)^2} D^{- \frac{(k-1)^2}{2}} k^{-(k-1)}. \eeq
\end{lemma}

\Lem~\ref{lemma_hessian_wormald} and its proof are very similar to an argument used in~\cite[\Sec~3]{WormaldColoring}.
In fact, \Lem~\ref{lemma_hessian_wormald} follows from

 \begin{lemma}[{\cite[\Lem~6 (b) and 7 (c)]{WormaldColoring}}]\label{Lemma_WormaldMatrix}
There is a $(k-1)^2\times(k-1)^2$-matrix $\cH=(\cH_{(i,j),(k,l)})_{i,j,k,l\in[k-1]}$ such that for any $\eps=(\eps_{ij})_{i,j\in[k]}\in\cS_n$ we have
	$$\sum_{i,j,i',j'\in[k-1]}\cH_{(i,j),(i',j')}\eps_{ij}\eps_{i'j'}=\norm\eps_2^2.$$
This matrix $\cH$ is positive definite and $\det\cH = k^{2(k-1)} $.
 \end{lemma}

 \begin{proof}[Proof of \Lem~\ref{lemma_hessian_wormald}] 
Together with the Euler-Maclaurin formula and \Lem~\ref{Lemma_WormaldMatrix}, a Gaussian integration yields
 \begin{align*}
 	\sum_{\epsilon \in S_n} \exp &\left[ - n \frac{D}{2} \| \epsilon\|_2^2 + o(n^{1/2}) \| \epsilon\|_2 \right]=
	\sum_{\epsilon \in \left( \mathbb{Z}/n \right)^{(k-1)^2} }
		\exp \left[ - n \frac{D}{2}\sum_{i,j,i',j'\in[k-1]}\cH_{(i,j),(i',j')}\eps_{ij}\eps_{i'j'} + o(n^{1/2}) \| \epsilon \|_2 \right]  \\ 
		& \sim n^{(k-1)^2} \int \dots  \int 
		\exp \left[ - n \frac{D}{2}\sum_{i,j,i',j'\in[k-1]}\cH_{(i,j),(i',j')}\eps_{ij}\eps_{i'j'}
	\right] 
			\mathrm d\eps_{11}\cdots\mathrm d\eps_{(k-1)(k-1)}
 \\ & \sim \left( \sqrt{2 \pi n} \right)^{(k-1)^2} D^{\frac{-(k-1)^2}{2}} (\det\cH)^{-1/2} 
	\sim \left( \sqrt{2 \pi n} \right)^{(k-1)^2} D^{\frac{-(k-1)^2}{2}} k^{-(k-1)} %
 , \end{align*} 
as desired.
  \end{proof}

\begin{proof}[Proof of \Prop~\ref{prop_second_moment_balanced_exact}] For  $\rho^{(1)}, \rho^{(2)} \in \cB_{n,k}(\w)$, we introduce the set of overlap matrices 
$$\cR_{n,k}^{\mathrm{bal}} (\w,n^{-5/12}, \rho^{(1)}, \rho^{(2)}) = \{ \rho \in \cR_{n,k}^{\mathrm{bal}}(\w, n^{-5/12}): \rho_{\nix\star}= \rho^{(1)}, 
		\rho_{\star\nix} = \rho^{(2)} \}.$$
In particular, $\cR_{k,n}^\bal(\w, n^{-5/12}, \rho^{(1)}, \rho^{(2)})$ contains the ``product'' overlap $\rho^{(1)} \otimes \rho^{(2)}$ defined by
$ ( \rho^{(1)} \otimes \rho^{(2)})_{ij} = \rho^{(1)}_i \rho^{(2)}_j . $
Because $\rho^{(1)}$ and $\rho^{(2)}$ are $(\w,n)$-balanced, we find
\beq \label{eq_aux_v_balanced_norm} \| \rho^{(1)} \otimes \rho^{(2)} - \bar{\rho} \|_2 = o(n^{-1/2}). \eeq
With these definitions we see that
\beq \label{eq_aux_v_1} \Erw \left[ Z_{k,\w, n^{-5/12}}^{(2)}(\cG(n,m)) \right] = \sum_{ \rho^{(1)} \in \cB_{n,k}(\w)} \sum_{ \rho^{(2)} \in \cB_{n,k}(\w)} \sum_{\rho \in \cR_{n,k}^{\bal}(\w,n^{-5/12},\rho^{(1)}, \rho^{(2)})} \Erw \left[ Z_{k,\rho}^{(2)}(\cG(n,m)) \right]. \eeq

Let us fix from now on two $(\w,n)$-balanced colour densities $\rho^{(1)}, \rho^{(2)}$ and simplify the notation by writing 
$$ \widehat{\cR} = \cR_{n,k}^{\mathrm{bal}} (\w,n^{-5/12}, \rho^{(1)}, \rho^{(2)}), \qquad \widehat{\rho} = \rho^{(1)} \otimes \rho^{(2)}. $$
Thus, we are going to evaluate $$\Sigma_1=\sum_{\rho \in \widehat{\cR}} \Erw \left[ Z_{k,\rho}^{(2)}(\cG(n,m)) \right].$$
 Eq. (\ref{eq_aux_bound_Z_rho_1}) of \Lem~\ref{lemma_expansion_saddle} gives
\begin{align} \label{eq_aux_v_add_1}  
\Sigma_1
 & \sim  \sum_{\substack{\rho \in \widehat{\cR}
 	 }} C_n(d,k)  \exp \left[ 2 n \alpha(d,k) - n \frac{D(d,k)}{2} \| \rho- \bar{\rho} \|_2^2 \right]. \end{align}
Further, by the triangle inequality,
 \beq \label{eq_aux_v_triangle}  \| \rho- \widehat{\rho} \|_2 - \| \widehat{\rho}  - \bar{\rho} \|_2 \leq \| \rho - \bar{\rho} \|_2 \leq  \| \rho- \widehat{\rho}\|_2 + \| \widehat{\rho}  - \bar{\rho} \|_2  .\eeq
 Along with (\ref{eq_aux_v_balanced_norm}) this gives $\| \rho - \bar{\rho} \|_2^2 = \| \rho- \widehat{\rho} \|_2^2  + o(n^{-1/2}) \| \rho - \widehat{\rho} \|_2  + o(n^{-1}). $
 Hence by replacing in (\ref{eq_aux_v_add_1}) we obtain with the notations of \Lem~\ref{lemma_expansion_saddle}
\begin{align}  \nonumber \Sigma_1&
	\sim \sum_{\substack{\rho \in \widehat{\cR}}} C_n(d,k) \exp \left[ 2 n \alpha(d,k) - n \frac{D(d,k)}{2} \| \rho- \widehat{\rho} \|_2^2 \right. 
			\left. + o(n^{1/2}) \| \rho- \widehat{\rho} \|_2 + o(1) \right]\\ 
	&  \sim C_n(d,k) \exp \left[  2 n \alpha(d,k) \right] \sum_{\substack{\rho \in \widehat{\cR}}}  
		\exp \left[- n \frac{D(d,k)}{2} \| \rho-\widehat{\rho} \|_2^2 \right.  \left.\phantom{\frac{D}{2}}+o(n^{1/2}) \| \rho- \widehat{\rho} \|_2  \right]. \label{eq_aux_v_3}
\end{align}
Moreover, with $\cS_n$ as in (\ref{eq_def_cSn}), it follows from (\ref{eq_aux_v_triangle}) that 
$$\left \{ \widehat{\rho} + \epsilon:  \epsilon \in \cS_n, \| \epsilon \|_2  \leq {n^{-5/12}}/2  \right \} \subset 
	\left \{ \rho \in \widehat{\cR}: \| \rho-\bar{\rho} \|_2 \leq n^{-5/12}  \right \} 
	\subset \left \{ \widehat{\rho} + \epsilon:  \epsilon \in \cS_n \right  \}. $$
Hence, 
 \begin{eqnarray*}
 \Sigma_2&= & C_n(d,k) \exp \left[  2 n \alpha(d,k) \right]
 \sum_{\substack{ \epsilon \in \cS_n \\ \| \epsilon\|_2 > n^{-5/12}/2}} \exp \left[ - n \frac{D(d,k)}{2} \| \epsilon \|_2^2 (1+o(1))\right]
  \\
 &=&  C_n(d,k) \exp \left[  2 n \alpha(d,k) \right]\sum_{\substack{l \in \mathbb{Z}/n \\ l > n^{-5/12}/2}} \sum_{ \substack{\epsilon \in S_n \\ \| \epsilon\|_2 = l}}  \exp \left[ - n l^2 \frac{D(d,k)}{2} (1+o(1))\right]  \\ 
 &= & C_n(d,k) \exp \left[  2 n \alpha(d,k) \right]O \left(n^{k^2} \right) \exp \left[ - \frac{D(d,k)}{2} n^{1/6} \right] . 
 \end{eqnarray*}
Consequently, (\ref{eq_aux_v_3}) yields $\Sigma_2=o(\Sigma_1)$.
Thus, we obtain
from \Lem~\ref{lemma_hessian_wormald} 
 that
	\begin{align}
	\Sigma_1&\sim\nonumber
	C_n(d,k) \exp \left[  2 n \alpha(d,k) \right]\sum_{\epsilon \in \cS_n} \exp \left[ - n \frac{D(d,k)}{2} \| \epsilon\|_2^2 + o(n^{-1/2}) \| \epsilon\|_2 \right].  \\
	 	&\sim C_n(d,k) \exp \left[  2 n \alpha(d,k) \right] \left( \sqrt{2\pi n}\right)^{(k-1)^2}k^{-k(k-1)}\left(1-\frac{d}{(k-1)^2} \right)^{-\frac{(k-1)^2}{2}}. \label{eqVV91}
	\end{align}
In particular, the last expression  is independent of the choice of the vectors $\rho^1, \rho^2$ that defined $\widehat{\cR}$.
Therefore, substituting~(\ref{eqVV91}) in the decomposition (\ref{eq_aux_v_1}) completes the proof of \Prop~\ref{prop_second_moment_balanced_exact}.
 \end{proof}

 \begin{proof}[Proof of \Prop s~\ref{prop_second_moment_bal_vanilla} and \ref{prop_first_moment_tame_bal}] \label{subsec_proof_final}
 First observe that $$ \exp \left (\sum_{l \geq 2} \lambda_l \delta_l^2 \right) = \left( 1- \frac{d}{(k-1)^2}\right)^{-\frac{(k-1)^2}{2}} \exp \left( - \frac{d}{2} \right) .$$
 \Prop~\ref{prop_second_moment_bal_vanilla} is immediately obtained by combining \Prop~\ref{prop_first_moment_balanced} with \Prop s~\ref{prop_ach_naor},~\ref{prop_second_moment_new_v} and \ref{prop_second_moment_balanced_exact}.
On the other hand, \Prop~\ref{prop_first_moment_tame_bal} is obtained by combining \Prop~\ref{prop_first_moment_balanced} with \Prop s~\ref{prop_aco_vil},~\ref{prop_second_moment_new_v} and \ref{prop_second_moment_balanced_exact}.
\end{proof}

\subsection{Proof of \Prop~\ref{prop_ach_naor}}\label{Sec_prop_ach_naor}
Let 
	\begin{equation}\label{eqf}
	f:\rho\in\overline{\cR}_k\ra\RR,\quad\rho\mapsto H(\rho) + \frac{d}{2} \ln\bc{ 1 - \frac{2}{k} + \norm\rho_2^2 }.
	\end{equation}
The following is a consequence of  Fact~\ref{fact_Z_rho_bal}.

\begin{fact}\label{fact_upper_lower_bounds_Z_rho}
Let $k \geq 3$, $d \in (0, \infty)$ and $\rho\in\cR_{n,k}^\bal(\w)$.
Then
	$\Erw [ Z_{k,\rho}^{(2)} (\gnm) ]=\exp(nf(\rho)+O(\ln n)).$
\end{fact}

\noindent
Fact~\ref{fact_upper_lower_bounds_Z_rho} reduces our task to studying the function $f(\rho)$.
For the range of $d$ covered by \Prop~\ref{prop_ach_naor}, this analysis is the main technical achievement of~\cite{AchNaor},
where (essentially) the following statement is proved.

\begin{lemma}\label{Lemma_AchNaorMain}
Assume that $k \geq 3$ and that $d \leq 2(k-1) \ln (k-1)$. For any $n >0$ and any $(\w,n)$-balanced overlap matrix $\rho$ we have
	\begin{equation}\label{eqAchNaorMain}
	f(\rho) \leq f(\bar{\rho}) - \frac{2(k-1) \ln (k-1) - d}{4(k-1)^2}  \left( k^2 \| \rho \|_2^2 -1 \right) + o(1). 
	\end{equation}	
\end{lemma}
\begin{proof}
For $\rho$ such that $\sum_{i=1}^k \rho_{ij} = \sum_{i=1}^k \rho_{ji} = 1/k$
the bound~(\ref{eqAchNaorMain}) is proved in \cite[\Sec~3]{AchNaor}.
This implies that~(\ref{eqAchNaorMain}) also holds for $\rho\in\cR_{n,k}^\bal(\w)$, because $f$ is uniformly continuous on the compact set $\overline\cR_k.$
\end{proof}

Now, assume that $k$ and $d$ satisfy the assumptions of \Prop~\ref{prop_ach_naor} and let $\eta >0$ be any fixed number.
The function $\overline{\cR} \to \mathbb{R}$,  $\rho \to k^2 \| \rho\|_2$ is smooth, strictly convex and attains its global minimum of $1$ at $\rho = \bar{\rho}$.
Consequently, there exist $c_k >0$ such that if $\| \rho - \bar{\rho} \|_2 > \eta$, then $\left(k^2\|\rho\|_2-1\right) \geq c_k$.
Hence, Fact~\ref{fact_upper_lower_bounds_Z_rho} and \Lem~\ref{Lemma_AchNaorMain} yield
	\begin{align}\label{eqAchNaor11}
	\sum_{ \substack{\rho \in \cR_{n,k}^\bal(\w) \\ \| \rho - \bar{\rho} \|_2 > \eta  }}
		\Erw \left[ Z_{k,\rho}^{(2)}(\gnm) \right] \leq \exp \left[ n f(\bar{\rho}) - n c_k d_k + o(n) \right],\quad\mbox{where }
			d_k = \frac{2(k-1) \ln (k-1)- d}{4 (k-1)^2} > 0.
	\end{align}
On the other hand, fixing any $\rho_0 \in \cR_{n,k}^\bal(\w)$ such that $\| \rho_0- \bar{\rho}\|_2 \leq k/n$, we obtain from
	Fact~\ref{fact_upper_lower_bounds_Z_rho} that
	\begin{align}\label{eqAchNaor12}
	\sum_{ \substack{\rho \in \cR_{n,k}^\bal(\w) \\ \| \rho - \bar{\rho} \|_2 \leq \eta  }} \Erw \left[ Z_{k,\rho}^{(2)} (\gnm) \right] \geq  
	\Erw \left[ Z_{k,\rho_0}^{(2)} (\gnm) \right] \geq 
		\exp \left[ n f(\bar{\rho}) +O(\ln n)\right]. \end{align}
Combining~(\ref{eqAchNaor11}) and~(\ref{eqAchNaor12}), we conclude that
$\Erw[ Z_{k,\w}^2(\gnm)] \sim \Erw[ Z_{k,\w, \eta}^{(2)}(\gnm)]$, 
thereby completing  the proof of \Prop~\ref{prop_ach_naor}.

\subsection{Proof of \Prop~\ref{prop_aco_vil}}\label{Sec_prop_aco_vil}
We continue to let $f$ denote the function from~(\ref{eqf}).
Let $\cB$ be the set of all $\rho\in\overline\cR_k$ such that
	$$\sum_{j=1}^k\rho_{ij}=\sum_{j=1}^k\rho_{ji}=1/k\quad\mbox{for all }i\in[k].$$
Further, let us say that $\rho\in\overline\cR_k$ is {\em $s$-stable} if $\rho$ has precisely $s$ entries in the interval $(0.51/k,1]$.
Then any $\rho\in\cB$ is $s$-stable for some $s\in\{0,1,\ldots,k\}$.
In addition,  let $\kappa=\ln^{20}k/k$ and
let us call $\rho\in\overline\cR_k$ {\em separable} if $k\rho_{ij}\not\in(0.51,1-\kappa)$ for all $i,j\in[k]$.
The following lemma summarizes the analysis of the function $f$ performed in \cite[\Sec~4]{Danny}.

\begin{lemma}\label{Lemma_Danny}
For any $c>0$ there is $k_0>0$ such that for all $k>k_0$ 
and all $d$ such that $(2k-1)\ln k-c\leq d\leq(2k-1)\ln k$ the following statements are true.
\begin{enumerate}
\item If $1\leq s<k$, then for all separable $s$-stable $\rho\in\cB$ we have $f(\rho)<f(\bar\rho)$.
\item If $\rho\in\cB$ is $0$-stable and $\rho\neq\bar\rho$, then $f(\rho)<f(\bar\rho)$.
\item If $d=(2k-1)\ln k-2$, then for all separable, $k$-stable $\rho\in\cB$ 
		we have $f(\rho)<f(\bar\rho)$.
\end{enumerate}
\end{lemma}

Further, let us call a $k$-colouring $\sigma$ of a graph $G$ on $[n]$ {\em separable}
if for any other $k$-colouring $\tau$ of $G$ the overlap matrix $\rho(\sigma,\tau)$ is separable.
The following is implicit in \cite[\Sec~3]{Danny}.

\begin{lemma}
	\label{Lemma_separable}
There is $k_0>0$ such that for all $k>k_0$ 
and all $d$ such that $2(k-1)\ln(k-1)\leq d\leq(2k-1)\ln k$ the following is true.
Let $\bar Z_{k,\omega}(\gnm)$ denote the number of $(\omega,n)$-balanced $k$-colourings of $\gnm$ that fail to be separable.
Then
	$\Erw[\bar Z_{k,\omega}(\gnm)]=o(\Erw[Z_{k,\omega}(\gnm)])$.
\end{lemma}

To state the final ingredient to the proof of \Prop~\ref{prop_aco_vil}, we need the following definition.
For a graph $G$ on $[n]$ and a $k$-colouring $\sigma$ of $G$ we let
	$\cC(G,\sigma)$ be the set of all $\tau\in\Bal$ that are $k$-colourings of $G$
such that $\rho(\sigma,\tau)$ is $k$-stable.

\begin{lemma}[{\cite{Cond}}]	\label{Lemma_ClusterSize}
There is $k_0>0$ such that for all $k>k_0$ 
and all $d$ such that $(2k-1)\ln k-2\leq d\leq\dc$ the following is true.
Let $\tilde Z_{k,\omega}(\gnm)$ denote the number of $(\omega,n)$-balanced $k$-colourings  such that
$|\cC(\gnm,\sigma)|>\Erw[Z_{k,\omega}(\gnm)]/n$.
Then
	$\Erw[\tilde Z_{k,\omega}(\gnm)]=o(\Erw[Z_{k,\omega}(\gnm)])$.
\end{lemma}

\begin{proof}[Proof of~\Prop~\ref{prop_aco_vil}]
Assume that $k\geq k_0$ for a large enough number $k_0$ and that $d\geq2(k-1)\ln(k-1)$.
We consider two different cases.
\begin{description}
\item[Case 1: $d\leq(2k-1)\ln k-2$]
 let $\Ztame$ be the number of $(\w,n)$-balanced separable $k$-colourings of $\gnm$.
 	Then \Lem~\ref{Lemma_separable} implies that 
	$ \Erw[ \Ztame(\gnm) ] \sim \Erw \left[ \Zbal(\gnm) \right]$.
	Furthermore, in the case that $d=(2k-1)\ln k-2$, the second and the third statement of \Lem~\ref{Lemma_Danny} imply that
	$f(\rho)<f(\bar\rho)$ for any separable $\rho\in\cB\setminus\cbc{\bar\rho}$.
	Because $f(\rho)$ is the sum of the concave function $\rho\mapsto H(\rho)$ and the convex function $\rho\mapsto\frac d2\ln(1-2/k\norm\rho_2^2)$,
	this implies that, in fact,
	for any $d\leq(2k-1)\ln k-2$ we have  $f(\rho)<f(\bar\rho)$ for any separable $\rho\in\cB\setminus\cbc{\bar\rho}$.
	Hence, the uniform continuity of $f$ on $\overline\cR_k$ 
	and Fact~\ref{fact_upper_lower_bounds_Z_rho} yield
		\begin{align}\label{eqprop_aco_vil_1}
		\Erw[\Ztame(\gnm)^2]\leq(1+o(1))
			\sum_{ \substack{\rho \in \cR_{n,k}^\bal(\w)\\\rho\mbox{\ \scriptsize is $0$-stable}}}
				\Erw \left[ Z_{k,\rho}^{(2)}(\gnm) \right].
		\end{align}
	Finally, combining~(\ref{eqprop_aco_vil_1}) with Fact~\ref{fact_upper_lower_bounds_Z_rho} and the third part of \Lem~\ref{Lemma_Danny}, we see that for any $\eta>0$,
			\begin{align}\label{eqprop_aco_vil_2}
			\sum_{ \substack{\rho \in \cR_{n,k}^\bal(\w)\\\rho\mbox{\ \scriptsize is $0$-stable}\\\norm{\rho-\bar\rho}_2>\eta}}
				\Erw \left[ Z_{k,\rho}^{(2)}(\gnm) \right]&\leq
					\sum_{ \substack{\rho \in \cR_{n,k}^\bal(\w)\\\rho\mbox{\ \scriptsize is $0$-stable}\\\norm{\rho-\bar\rho}_2>\eta}}
				\exp(nf(\rho)+O(\ln n))=o\bc{\Erw [Z^{(2)}_{k, \w, \eta}(\cG(n,m))]}.
		\end{align}
	The assertion follows by combining~(\ref{eqprop_aco_vil_1}) and~(\ref{eqprop_aco_vil_2}).
\item[Case 2: $(2k-1)\ln k-2<d<\dc$]
	 let $\Ztame$ be the number of $(\w,n)$-balanced separable $k$-colourings $\sigma$ of $\gnm$
	 such that $|\cC(\gnm,\sigma)|\leq\Erw[Z_{k,\omega}(\gnm)]/n$.
	 Then \Lem s~\ref{Lemma_separable} and~\ref{Lemma_ClusterSize} imply that
	 		$ \Erw[ \Ztame(\gnm) ] \sim \Erw \left[ \Zbal(\gnm) \right]$.
	Furthermore, the 
	first part of \Lem~\ref{Lemma_Danny} and Fact~\ref{fact_upper_lower_bounds_Z_rho} entail that~(\ref{eqprop_aco_vil_1}) holds for this random variable $\Ztame$.
	Moreover, as in the previous case
	(\ref{eqprop_aco_vil_1}), Fact~\ref{fact_upper_lower_bounds_Z_rho} and the third part of \Lem~\ref{Lemma_Danny}
	show that~(\ref{eqprop_aco_vil_2}) holds true for any fixed $\eta>0$.
\end{description}
In either case the assertion follows by combining~(\ref{eqprop_aco_vil_1}) and~(\ref{eqprop_aco_vil_2}).
\end{proof}

\end{document}